\documentclass{lmcs} 
\pdfoutput=1

\usepackage{lastpage}
\lmcsdoi{19}{1}{10}
\lmcsheading{}{\pageref{LastPage}}{}{}%
{Mar.~18,~2022}{Feb.~08,~2023}{}


\usepackage{hyperref}
\usepackage{amsmath}
\usepackage{amssymb}
\usepackage{amsthm}
\usepackage{amscd}
\usepackage{amsfonts}
\usepackage{float}
\usepackage{mathtools}
\usepackage{xspace}
\usepackage[utf8]{inputenc}
\usepackage{color, colortbl}
\usepackage{algorithm}
\usepackage[noend]{algpseudocode}
\usepackage{graphicx}
\usepackage{qtree}
\usepackage{multicol}
\usepackage{booktabs}
\usepackage{wrapfig}
\usepackage{tikz}
\usetikzlibrary{calc}
\usetikzlibrary{automata, positioning, arrows, shapes.geometric}
\usepackage[font=small,labelfont=bf]{caption}
\usepackage{stmaryrd} 

\newcommand\diagfil[4]{%
  \multicolumn{1}{p{#1}}{\hskip-\tabcolsep
  $\vcenter{\begin{tikzpicture}[baseline=0,anchor=south west,inner sep=0pt,outer sep=0pt]
  \path[use as bounding box] (0,0) rectangle (#1+2\tabcolsep,1.23\baselineskip);
  \node[minimum width={#1+2\tabcolsep},minimum height=\baselineskip+\extrarowheight+\belowrulesep+\aboverulesep,fill=#2] (box)at(0,-\aboverulesep) {};
  \fill [#3] (box.south west)--(box.north east)|- cycle;
  \node[anchor=center] at (box.center) {#4};
  \draw (#1+1.97\tabcolsep,0) -- (#1+1.97\tabcolsep,1.23\baselineskip);
  \end{tikzpicture}}$\hskip-\tabcolsep}}

\newcommand{\commentout}[1]{}

\newcommand{\dana}[1]{\colorbox{yellow}{......}\footnote{\colorbox{yellow}{\textsc{df: [}}#1{\colorbox{yellow}{]}}}}

\newcommand{\reals}{\mathbb{R}}

\newcommand{\biword}[2]{\ensuremath{\left(\begin{smallmatrix} #1 \\ #2 \end{smallmatrix}\right)}}

\newcommand{\grmrwgt}[2]{\ensuremath{\mathcal{W}_{#1}(#2)}}

\newcommand{\query}[1]{\textsc{#1}}
\newcommand{\mq}{\query{mq}}
\newcommand{\smq}{\query{smq}}
\newcommand{\eq}{\query{eq}}
\newcommand{\seq}{\query{seq}}

\newcommand{\complexityclass}[1]{\textbf{#1}}
\newcommand{\RP}{\complexityclass{RP}}
\newcommand{\NP}{\complexityclass{NP}}

\newcommand{\staralg}[1]{\ensuremath{{\textbf{#1}^*}}}
\newcommand{\lstar}{\staralg{L}}

\newcommand{\gstar}{\staralg{G}}
\newcommand{\mstar}{\staralg{M}}

\newcommand{\cstar}{\staralg{C}}

\newcommand{\deriv}{\ensuremath{\mathsf{T}}}
\newcommand{\struct}{\ensuremath{\mathsf{S}}}
\newcommand{\skels}{\ensuremath{\mathsf{S}}}
\newcommand{\skel}{\ensuremath{\mathsf{S}}}
\newcommand{\skelstotrees}{\ensuremath\mathsf{T}}

\newcommand{\aut}[1]{\mathcal{#1}}
\newcommand{\grmr}[1]{\ensuremath{\mathcal{#1}}}
\newcommand{\treesrs}[1]{\ensuremath{\mathcal{#1}}}

\newcommand{\treesEquiv}[1]{\equiv_{#1}}
\newcommand{\treesColin}[2]{\ltimes^{#1}_{#2}}

\newcommand{\classrepr}[2]{[#1]}
\newcommand{\classcoeff}[2]{\alpha[#1]}
\newcommand{\classind}[2]{\iota[#1]}

\newcommand{\tuple}[1]{\ensuremath\langle #1 \rangle}

\newcommand{\derives}{\xrightarrow{}}
\newcommand{\transderives}[1]{\ensuremath{\Rightarrow_{#1}}}
\newcommand{\Vars}{\ensuremath{\mathcal{V}}}

\newcommand{\Prob}{\mathbb{P}}

\theoremstyle{plain} \numberwithin{equation}{section}
\newtheorem{theorem}[thm]{Theorem}
\newtheorem{proposition}[theorem]{Proposition}
\newtheorem{corollary}[theorem]{Corrolary} 
\newtheorem{lemma}[theorem]{Lemma} 

\newcommand{\V}{\mathbb{V}}
\newcommand{\K}{\mathbb{K}}
\newcommand{\R}{\mathbb{R}}

\newcommand{\hm}[1]{\ensuremath{{#1}}}

\newcommand{\crank}{\ensuremath{\textsl{c-rank}_{+}}}
\newcommand{\rrank}{\ensuremath{\textsl{r-rank}_{+}}}
\newcommand{\prank}{\ensuremath{\textsl{rank}_{+}}}
\newcommand{\trees}{\textsl{Trees}}

\newcommand{\pos}{\textsl{span}_{+}}
\newcommand{\posspan}[1]{\pos(#1)}

\newcommand{\proc}[1]{\textsl{#1}}

\newcommand{\ExtractCMTA}{\proc{ExtractCMTA}}
\newcommand{\LearnCMTA}{\proc{LearnCMTA}}
\newcommand{\Close}{\proc{Close}}
\newcommand{\Consistent}{\proc{Consistent}}
\newcommand{\Complete}{\proc{Complete}}

\newcommand{\AlgCloseCMTA}{2\xspace}
\newcommand{\AlgConsistentCMTA}{3\xspace}
\newcommand{\AlgCompleteCMTA}{4\xspace}
\newcommand{\AlgExtractCMTA}{5\xspace}
\newcommand{\AlgLearnCMTA}{1\xspace}

\newcommand{\sema}[1]{{\llbracket}#1{\rrbracket}}%
\newcommand{\semagrmr}[1]{#1}%

\newcommand{\context}{\diamond}

\newcommand{\pref}{\textsl{Pref}}

\newcommand{\contextConcat}[2]{#1\llbracket#2\rrbracket}

\theoremstyle{plain} 

\begin{document}

\title[Learning of Structurally Unambiguous Probabilistic Grammars]{Learning of Structurally Unambiguous\texorpdfstring{\\}{ }Probabilistic Grammars}

\author[D.~Fisman]{Dana Fisman}	
\author[D.~Nitay]{Dolav Nitay}	
\author[M.~Ziv-Ukelson]{Michal Ziv-Ukelson}	

\address{Ben-Gurion University, Israel}	
\email{dana@cs.bgu.ac.il, dolavn@post.bgu.ac.il, michaluz@cs.bgu.ac.il}  





\begin{abstract}
  \noindent 	The problem of identifying a probabilistic context free grammar has two aspects: the first is determining the grammar's topology (the rules of the grammar) and the second is estimating probabilistic weights for each rule. Given the hardness results for learning context-free grammars in general, and probabilistic grammars in particular, most of the literature has concentrated on the second problem. In this work we address the first problem.  We restrict attention to \emph{structurally unambiguous weighted context-free grammars} (SUWCFG) and provide a query learning algorithm for \emph{structurally unambiguous probabilistic context-free grammars}  (SUPCFG). We show that SUWCFG can be represented using \emph{co-linear multiplicity tree automata} (CMTA), and provide a polynomial learning algorithm that learns CMTAs.  We show that the learned CMTA can be converted into a probabilistic grammar, thus providing  a complete algorithm for learning  a structurally unambiguous probabilistic context free grammar (both the grammar topology and the probabilistic weights) using structured membership queries and structured equivalence queries. 
  A summarized version of this work was published at AAAI 21~\cite{NitayFZ21}.
\end{abstract}

\maketitle


\section{Introduction}
Probabilistic context free grammars (PCFGs) constitute a computational model suitable for probabilistic systems which observe non-regular (yet context-free) behavior. They are vastly used in computational linguistics~\cite{Chomsky56}, natural language processing~\cite{church-1988-stochastic} and biological modeling, for instance, in probabilistic modeling of RNA structures~\cite{Grate95}.  Methods for learning PCFGs from experimental data have been studied for over half a century.
Unfortunately, there are various hardness results regarding learning context-free grammars in general and probabilistic grammars in particular. 
It follows from~\cite{Gold78}  that context-free grammars (CFGs) cannot be identified in the limit from positive examples, and from~\cite{Angluin90}  that CFGs cannot be identified in polynomial time using equivalence queries only. Both results are not surprising for those familiar with learning regular languages, as they hold for the class of regular languages as well. However, while regular languages can be learned using both membership queries and equivalence queries~\cite{Angluin87}, it was shown that learning CFGs using both membership queries and equivalence queries is computationally as hard as key cryptographic problems for which there is currently no known polynomial-time algorithm~\cite{AngluinK95}. de la Higuera elaborates more on  the difficulties of learning context-free grammars in his book~\cite[Chapter 15]{delaHiguera}.
Hardness results for the probabilistic setting have also been established. Abe and Warmuth have shown a computational hardness result for the inference of probabilistic automata, in particular, that an exponential blowup with respect to the alphabet size is inevitable unless $\RP = \NP$~\cite{AbeW92}. 

The problem of identifying a probabilistic grammar from examples has two aspects: the first is determining the rules of the grammar up to variable renaming and the second is estimating probabilistic weights for each rule. Given the hardness results mentioned above, most of the literature has concentrated on the second problem. Two dominant approaches for solving the second problem are the forward-backward algorithm for HMMs~\cite{Rabiner89} and the inside-outside algorithm for PCFGs~\cite{Baker79,LariY90}.

In this work we address the first problem. Due to the hardness results regarding learning probabilistic grammars using \emph{membership queries} and \emph{equivalence queries} (\mq\ and \eq) we use \emph{structured membership queries} and \emph{structured equivalence queries} (\smq\ and \seq), as was done by~\cite{Sakakibara88} for learning context-free grammars. \emph{Structured strings}, proposed by~\cite{LevyJ78}, are strings over the given alphabet that includes parentheses that indicate the structure of a possible derivation tree for the string. One can equivalently think about a structured string as a derivation tree in which all nodes but the leaves are marked with $?$, namely an \emph{unlabeled derivation tree}. 

It is known that the set of derivation trees of a given CFG constitutes a \emph{regular tree-language}, where a regular tree-language is a tree-language that can be recognized by a \emph{tree automaton} \cite{LEVY1978192}.
Sakakibara  has generalized Angluin's \lstar\ algorithm (for learning regular languages using \mq\ and \eq) to learning a tree automaton, and provided a polynomial learning algorithm for CFGs using \smq\ and \seq~\cite{Sakakibara88}.
Let $\deriv(\grmr{G})$ denote the set of derivation trees of a CFG $\grmr{G}$, and $\struct(\deriv(\grmr{G}))$ the set of unlabeled derivation trees (namely the structured strings of $\grmr{G}$). 
While a membership query (\mq) asks whether a given string $w$ is in the unknown grammar $\grmr{G}$, 
a structured membership query (\smq) asks whether a structured string $s$ is in   $\struct(\deriv(\grmr{G}))$ and a structured equivalence query (\seq) answers whether the queried CFG $\grmr{G}'$ is structurally equivalent to the unknown grammar $\grmr{G}$, and accompanies a negative answer with a structured string $s'$ in the symmetric difference of  $\struct(\deriv(\grmr{G}'))$ and $\struct(\deriv(\grmr{G}))$. 

In our setting, since we are interested in learning probabilistic grammars, an \smq\ on a structured string $s$ is answered by a weight $p\in[0,1]$ standing for the probability for $\grmr{G}$ to generate $s$, and a negative answer to an \seq\ is accompanied by a structured string $s$ such that $\grmr{G}$ and $\grmr{G}'$ generate $s$ with different probabilities, along with  the probability $p$ with which the unknown grammar $\grmr{G}$ generates $s$. 

Sakakibara works with tree automata to model the derivation trees of the unknown grammars~\cite{Sakakibara88}.
In our case the automaton needs to associate a weight with every tree (representing a structured string). We choose to work with the model of \emph{multiplicity tree automata}. A multiplicity tree automaton (MTA) associates with every tree a value from a given field $\K$. An algorithm for learning multiplicity tree automata, to which we refer as \mstar, was developed in~\cite{habrard2006learning,drewes2007query}.\footnote{Following a learning algorithm developed for multiplicity word automata~\cite{BergadanoV96,BeimelBBKV00}.} 

A probabilistic grammar is a special case of a weighted grammar and~\cite{abney1999relating,smith2007weighted} have shown that  convergent weighted CFGs (WCFGs) where all weights are non-negative and probabilistic CFGs (PCFGs) are equally expressive.\footnote{The definition of \emph{convergent} is deferred to the preliminaries.} We thus might expect to be able to use the learning algorithm \mstar\ to learn an MTA corresponding to a WCFG, and apply this conversion  to the result, in order to obtain the desired PCFG. However, as we show in  Proposition~\ref{prop:mstar_negative_mta}, there are probabilistic languages for which applying the \mstar\ algorithm results in an MTA with negative weights. Trying to adjust the algorithm to learn a positive basis may encounter the issue that for some PCFGs,  no finite subset of the infinite Hankel Matrix spans the entire space of the function, as we show in Proposition~\ref{prop:pcfg-no-finite-rank}.\footnote{The definition of the Hankel Matrix and its role in learning algorithms appears in the sequel.}  To overcome these issues we restrict attention to structurally unambiguous grammars (SUCFG, see section \ref{sec:sucfg}), which as we show, can be modeled using co-linear multiplicity automata (defined next).

We develop a polynomial learning algorithm, which we term \cstar, that learns a restriction of MTA, which we term \emph{co-linear multiplicity tree automata} (CMTA). We then show that a CMTA for a probabilistic language can be converted into a PCFG, thus yielding a complete algorithm for learning SUPCFGs using \smq s and \seq s as desired. 

A summarized version of this work was published at AAAI'21~\cite{NitayFZ21}.

\section{Preliminaries} 

This section provides the  definitions required for  \emph{probabilistic grammars} -- the object we design a learning algorithm for,
and \emph{multiplicity tree automata}, the object we use in the learning algorithm.
\subsection{Probabilistic Grammars}\label{subsec:pcfgs}

Probabilistic grammars are a special case of context free grammars 
where each production rule has a weight in the range $[0,1]$ and for each non-terminal, the sum of weights of its productions is one. 

	A \emph{context free grammar} (CFG) is a quadruple $\grmr{G}=\langle \Vars,\Sigma,R,S\rangle$, where
	$\Vars$ is a finite non-empty set of symbols called  \emph{variables} or \emph{non-terminals},
	$\Sigma$ is a finite non-empty set of symbols called the \emph{alphabet} or the \emph{terminals},
	$R\subseteq \Vars\times (\Vars\cup\Sigma)^{*}$ is a relation between variables and strings over $\Vars\cup\Sigma$, called the  \emph{production rules}, and $S\in \Vars$ is a special variable called the \emph{start variable}. 
	We assume the reader is familiar with the standard definition of CFGs and of derivation trees.

We say that $S\transderives{} w$ for a string $w\in\Sigma^*$ if there exists a derivation tree $t$ such that all leaves are in $\Sigma$ and when concatenated from left to right they form $w$. That is, $w$ is the \emph{yield} of the tree $t$. In this case we also use the notation $S\transderives{t} w$.
A CFG $\grmr{G}$ defines a set of words over $\Sigma$, the \emph{language generated by} $\grmr{G}$, which  is the set of words $w\in\Sigma^*$ such that $S\transderives{}  w$, and is denoted $\sema{\grmr{G}}$. For simplicity, we assume the grammar does not derive the empty word.

\paragraph{Weighted grammars}
A \emph{weighted grammar}  (WCFG) is a pair $\tuple{\grmr{G},\theta}$ where $\grmr{G}=\langle \Vars,\Sigma,R,S\rangle$ is a CFG and $\theta:R\rightarrow \R$ is a function mapping each production rule to a weight in $\R$.
A WCFG $\grmr{W}=\tuple{\grmr{G},\theta}$ defines a function from words over $\Sigma$ to weights in $\R$. 
The WCFG associates with a derivation tree $t$ its weight, which is defined as 
$$\grmr{W}(t)=\prod_{(V\derives \alpha)\in R } \theta(V\derives\alpha)^{\sharp_t(V\derives\alpha)}$$
where $\sharp_t(V\derives\alpha)$ is the number of occurrences of the production $V\derives\alpha$ in the derivation tree $t$.
We abuse notation and treat $\semagrmr{\grmr{W}}$ also as a function from $\Sigma^*$ to $\R$  defined as 
$\grmr{W}(w)=\sum_{S\transderives{t}w}\grmr{W}(t)$.
That is, the weight of $w$ is the sum of weights of the derivation trees yielding $w$, and if $w\notin\sema{\aut{G}}$ then $\semagrmr{\grmr{W}}(w)=0$.
If the sum of all derivation trees in $\sema{\grmr{G}}$, namely $\sum_{w\in\sema{\grmr{G}}}\grmr{W}(w)$, is finite we say that $\aut{W}$ is \emph{convergent}. 
Otherwise, we say that $\aut{W}$ is \emph{divergent}.

\paragraph{Probabilistic grammars}
A \emph{probabilistic grammar} (PCFG) is a WCFG 
$\grmr{P}=\tuple{\grmr{G},\theta}$ where $\grmr{G}=\langle \Vars,\Sigma,R,S\rangle$ is a CFG and $\theta:R\rightarrow [0,1]$ is a function mapping each production rule of $\grmr{G}$ to a weight in the range $[0,1]$ that satisfies 
$$1=\sum_{(V \derives \alpha_i) \in R} \theta(V\derives \alpha_i) $$ 
for every $V\in \Vars$.\footnote{Probabilistic grammars are sometimes called \emph{stochastic grammars (SCFGs)}.}
One can see that if $\grmr{P}$ is a PCFG then  the sum of all derivations equals $1$, thus $\grmr{P}$  is convergent.

\paragraph{Word and Tree Series} \label{sec:series}
While \emph{words} are defined as sequences over a given alphabet, trees are defined using a \emph{ranked alphabet},
an alphabet $\Sigma=\{\Sigma_0,\Sigma_1,\ldots,\Sigma_p\}$ which is a  tuple of alphabets $\Sigma_k$ where $\Sigma_0$ is non-empty.
Let $\trees(\Sigma)$ be the set of trees over $\Sigma$, where a node labeled $\sigma\in\Sigma_k$ for $0\leq k \leq p$ has exactly $k$ children. While a \emph{word language} is a function mapping all possible words (elements of $\Sigma^*$) to $\{0,1\}$, a \emph{tree language} is a function from all possible trees (elements of $\trees(\Sigma)$) to $\{0,1\}$. We are interested in assigning each word or tree a non-Boolean value, usually a weight $p\in[0,1]$. More generally, let $\K$ be an arbitrary domain, e.g. the real numbers. We are interested in functions mapping words or trees to values in $\K$. A function from $\Sigma^*$ to $\K$ is  called a \emph{word series}, and a function from $\trees(\Sigma)$ to $\K$ is referred to as a \emph{tree series}. CFGs define word languages, and induce tree languages (the parse trees deriving the words defined by the grammars).
WCFGs and PCFGs define word series, where in the latter case the map is from $\Sigma^*$ to $[0,1]$.

\subsection{Multiplicity Tree Automata}\label{sec:mta}

\paragraph{Word and Tree Automata}
\emph{Word automata} are machines that recognize word languages, i.e. they define a function from $\Sigma^*$ to $\{0,1\}$.
\emph{Tree automata} are machines that recognize tree languages, i.e. they define a function from $\trees(\Sigma)$ to $\{0,1\}$.
\emph{Multiplicity word automata} (MA) are machines to implement word series $f:\Sigma^*\rightarrow\K$ where $\K$ is a field.
\emph{Multiplicity tree automata} (MTA) are machines to implement tree series  $f:\trees(\Sigma)\rightarrow\K$ where $\K$ is a field.

\paragraph{Multiplicity Automata}
Multiplicity automata can be thought of as an algebraic extension of automata, in which reading an input letter is implemented by matrix multiplication.
In a multiplicity word automaton with dimension $m$ over alphabet $\Sigma$, for each $\sigma\in\Sigma$ there is an $m$ by $m$ matrix, $\mu_\sigma$, whose entries are values in $\K$ where intuitively the value of entry $\mu_\sigma(i,j)$ is the weight of the passage from state $i$ to state $j$. The definition of multiplicity tree automata is a bit more involved; it makes use of multilinear functions as defined next. 

\paragraph{Multilinear functions}\label{par:multilinear functions}
Let $\V=\K^d$ be the $d$ dimensional vector space over $\K$. Let $\eta : \V^k \rightarrow \V$ be a $k$-linear function.
We can represent $\eta$ by a $d$ by $d^k$ matrix over $\K$. 

\begin{exa}
For instance, if $\eta : \V^3 \rightarrow \V$ and $d=2$ (i.e. $\V = \K^2$) then $\eta$ can be represented by the $2\times 2^3$ matrix $M_\eta$ provided in Fig~\ref{fig:mult-aut-matrices} where $c^{i}_{j_1 j_2 j_3} \in \K$ for $i,j_1,j_2,j_3\in\{1,2\}$. 
Then $\eta$, a function taking $k$ parameters in $\V=\K^d$, can be computed by multiplying the matrix $M_\eta$ with a vector for the parameters for $\eta$. Continuing this example, given the parameters $\textbf{x} = (x_1 \ \ x_2)$, $\textbf{y} = (y_1 \  \ y_2)$, $\textbf{z} = (z_1\ \  z_2)$ the value
$\eta(\textbf{x},\textbf{y},\textbf{z})$ can be calculated using the multiplication $M_\eta P_{xyz}$ where the vector $P_{xyz}$  of size $2^3$ is provided in Fig~\ref{fig:mult-aut-matrices}.
\end{exa}
\begin{figure}
	\hspace{-3mm}
	\scalebox{.9}{\small{
                \begin{tabular}{ c@{\qquad}c}
			    $M_\eta$ = & $~~~P_{xyz} =~~~$ \\[1mm]
			    $\begin{pmatrix}
				c^1_{{111}} & c^1_{{112}}& c^1_{{121}} & c^1_{{122}} & c^1_{{211}} & c^1_{{212}}& c^1_{{221}} & c^1_{{222}} \\[1mm]
				c^2_{{111}} & c^2_{{112}}& c^2_{{121}} & c^2_{{122}} & c^2_{{211}} & c^2_{{212}}& c^2_{{221}} & c^2_{{222}} 
				\end{pmatrix}$
				& \qquad 
				$\begin{pmatrix}
				x_1 y_1 z_1 \\
				x_1 y_1 z_2 \\
				x_1 y_2 z_1 \\
				\ldots \\
				x_2 y_2 z_2 \\
				\end{pmatrix}$\phantom{--}
			\end{tabular}}
	}
	\caption{A matrix $M_\eta$ for a multi-linear function $\eta$ and a vector $P_{xyz}$ for the respective $3$ parameters.
	}\label{fig:mult-aut-matrices}
\end{figure} 
In general, if $\eta : \V^k \rightarrow \V$ is such that $\eta(\textbf{x}_1,\textbf{x}_2,\ldots,\textbf{x}_k) = \textbf{y}$ and $M_\eta$, the matrix representation of $\eta$,  is defined using the constants $c^{i}_{j_1 j_2 \ldots j_k}$ then 
\[\textbf{y}[i] =  \sum_{\left\{(j_1,j_2,\ldots,j_k)\in \{1,2,\ldots,d\}^k\right\}} \hspace{-1em} c^i_{j_1j_2\ldots j_k} \, \textbf{x}_1[j_1]\, \textbf{x}_2[j_2]\, \cdots \,\textbf{x}_k[j_k]\]

\paragraph{Multiplicity tree automata}
A \emph{multiplicity tree automaton} (MTA) is a tuple $\aut{M}=(\Sigma,\K,d,\mu,\lambda)$ where $\Sigma=\{\Sigma_0,\Sigma_1,\ldots,\Sigma_p\}$ is the given ranked alphabet, $\K$ is a field corresponding to the range of the tree-series, $d$ is a non-negative integer called the  automaton \emph{dimension}, $\mu$ and $\lambda$ are  the transition and output function, respectively, whose types are defined next.
Let $\V=\K^d$. Then  $\lambda$ is an element of $\V$, namely a $d$-vector over $\K$.  Intuitively, $\lambda$ corresponds to the final  values of the ``states'' of $\aut{M}$. The transition function $\mu$ maps each element $\sigma$  of $\Sigma$ to a dedicated transition function $\mu_\sigma$ such that given ${\sigma \in \Sigma_k}$ for ${0 \leq k \leq p}$ then $\mu_\sigma$ is a $k$-linear function from $\V^k$ to $\V$. 
The transition function $\mu$ induces a function from $\trees(\Sigma)$ to $\V$, defined as follows.
If $t=\sigma$ for some $\sigma\in \Sigma_0$, namely $t$ is a tree with one node which is a leaf, then $\mu(t)=\mu_\sigma$ (note that $\mu_\sigma$ is a vector in $\K^d$ when $\sigma\in\Sigma_0$). 
If $t=\sigma(t_1,\ldots,t_k)$, namely $t$ is a tree with root $\sigma\in \Sigma_k$ and children $t_1,\ldots,t_k$ then $\mu(t)=\mu_\sigma(\mu(t_1),\ldots,\mu(t_k))$. 
The automaton $\aut{M}$ induces a total function from $\trees(\Sigma)$ to $\K$ defined as follows: $\aut{M}(t)=\lambda\cdot \mu(t)$.

\begin{exa}
Fig.~\ref{fig:MTA} on the left provides an example of an MTA $\aut{M}=((\Sigma_0,\Sigma_2),\R,2,\mu,\lambda)$
		where    ${\Sigma_0=\{a\}}$ and ${\Sigma_2=\{b\}}$  implementing a tree series that returns the number of leaves in the tree. 
		Fig.~\ref{fig:MTA} on the right provides a tree where a node $t$ is annotated by $\mu(t)$. Since $\mu(t_\epsilon)=\biword{3}{1}$, 
		where $t_\epsilon$ is the root, the value of the entire tree is  $\lambda\cdot\biword{3}{1}=3$.
\end{exa}

\begin{figure} 
    \commentout
    {
	\scalebox{.8}{
	\noindent\makebox[\textwidth/2]{
    \hspace{10mm}\includegraphics[scale=0.99,page=8, clip, trim=2cm 22.5cm 9cm 1.8cm]{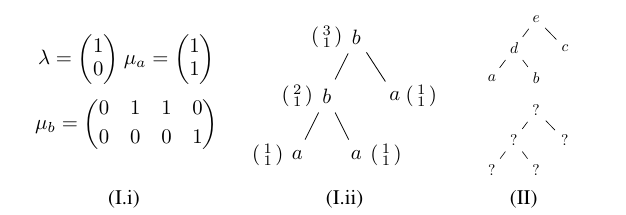}
    }}
    }
	\commentout
	{
		\begin{tabular}{cc@{\qquad}c}
		\includegraphics[scale=0.19]{MTA_def.png} 
		&
		\includegraphics[scale=0.19]{MTA_tree.png}
		&
		\includegraphics[scale=0.19]{structures_string.png}\\
		(I.i) & (I.ii) & (II)
		\end{tabular}
		}
		{
			\begin{tabular}{c@{\qquad\qquad\qquad}c}
			\begin{tabular}{l}
				$\lambda=\begin{pmatrix}1\\0
				\end{pmatrix}$ \ \ \ 
				$\mu_{a}= \begin{pmatrix}1\\1
				\end{pmatrix}$ 
				\\ \quad \\ \vspace{5mm}
				$\mu_{b} = \begin{pmatrix}
				0 & 1 & 1 & 0 \\[1mm]
				0 & 0 & 0 & 1 
				\end{pmatrix}$
				\end{tabular}
                &
                \begin{tabular}{l}
                \quad \\ \vspace{5mm}
				\begin{tikzpicture}
				\tikzstyle{myarrow}=[line width=.5mm,draw=#1,-triangle 45,postaction={draw, line width=4mm, shorten >=5mm, -}]
				\node[shape=circle, draw=none] (a1) at (2, 0) {$a$};
				\node[] (a1v) at (3.5, 0) {${\biword{1}{1}}$};
				\node[] (a2v) at (1.5, 0) {${\biword{1}{1}}$};
				\node[shape=circle, draw=none]  (a2) at (3, 0) {$a$};
				\node[shape=circle, draw=none] (a3) at (3.65, 1.0) {$a$};
				\node[] (a3v) at (4.10, 1.0) {${\biword{1}{1}}$};
				\node[shape=circle, draw=none]  (b1) at (2.5, 1.0) {$b$};
				\node[] (a3v) at (2.0, 1.0) {${\biword{2}{1}}$};
				\node[shape=circle, draw=none]  (b2) at (3, 2) {$b$};
				\node[] (b2v) at (2.5, 2) {${\biword{3}{1}}$};
				\draw (b2) -- (a3);
				\draw (b2) -- (b1);
				\draw (b1) -- (a1);
				\draw (b1) -- (a2);
				\end{tikzpicture} 
				\vspace{-5mm}
				\end{tabular}
	\end{tabular}}
	
	\caption{(Left) An MTA implementing a tree series that returns the number of leaves in the tree. 
		(Right) a tree where a node $t$ is annotated by $\mu(t)$. }\label{fig:MTA}
\end{figure}

\subsection{Contexts and Structured Trees}

\paragraph{Contexts}
When learning word languages, learning algorithms typically distinguish words $u$ and $v$ if there exists a suffix $z$ such that $uz$ is accepted but $vz$ is rejected or vice versa.
When learning tree languages, in a similar way we would like to distinguish between two trees $t_u$ and $t_v$,
if there exists a tree $t_z$ whose composition with $t_u$ and $t_v$ is accepted in one case and rejected in the other. 
To define this formally we need some notion to compose trees, more accurately we compose trees with \emph{contexts} as defined next.
Let $\Sigma=\{\Sigma_0,\Sigma_1,\ldots,\Sigma_p\}$ be a ranked alphabet. Let $\context$ be a symbol not in $\Sigma$. 
We use $\trees_\context(\Sigma)$ to denote all non-empty trees over $\Sigma' = \{ \Sigma_0 \cup \{\context\}, \Sigma_1,\ldots, \Sigma_p\}$  in which $\context$ appears exactly once.
Intuitively $\context$ indicates the place where a tree $t_u$ can be composed with a context $t_z$ yielding a tree that resembles $t_z$ but has $t_u$ as a sub-tree instead of the leaf $\context$.
We refer to an element of $\trees_\context(\Sigma)$ as a \emph{context}. Note that  at most one child of any node in a context $c$ is a context; the other ones are pure trees (i.e. elements of $\trees(\Sigma))$.
Given a tree ${t\in \trees(\Sigma)}$ and context $c\in \trees_\context(\Sigma)$ we use $\contextConcat{c}{t}$ for the tree $t'\in\trees(\Sigma)$ obtained from $c$ by replacing $\context$ with $t$.

\paragraph{Structured tree languages/series}
Recall that our motivation is to learn a word (string) series rather than a tree series, and due to hardness results on learning CFGs and PCFGs we resort to using \emph{structured strings}.  A \emph{structured string} is a string with parentheses exposing the structure of a derivation tree for the corresponding trees, as exemplified in Fig.~\ref{fig:structured-string}.

\begin{wrapfigure}{r}{0.2\textwidth}
    \vspace{-2mm}
	\scalebox{.8}{
    \includegraphics[scale=1.0,page=8, clip, trim=9.8cm 23.2cm 8.2cm 1.6cm]{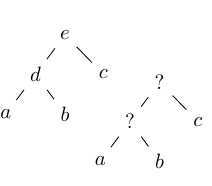}
    }
\caption{A derivation tree and its corresponding skeletal tree, which can be written as the structured string $((ab)c)$.}\label{fig:structured-string}
\end{wrapfigure} 

A \textit{skeletal alphabet} is a ranked alphabet in which we use a special symbol $?\notin \Sigma_0$ and for every $0<k\leq p$ the set $\Sigma_{k}$ consists only of the symbol $?$. For $t\in\trees(\Sigma)$, the skeletal description of $t$, denoted by $\skel(t)$, is a tree with the same topology as $t$, in which the symbol in all internal nodes is $?$, and the symbols in all leaves are the same as in $t$. Let $T$ be a set of trees. The corresponding skeletal set, denoted $\skels(T)$ is $\{\skels(t)~|~t\in T\}$.
Going from the other direction, given a skeletal tree $s$ we use $\skelstotrees(s)$ for the set $\{t\in \trees(\Sigma)~|~\skel(t)= s\}$.

A tree language over a skeletal alphabet is called a \textit{skeletal tree language}. And a mapping from skeletal trees to $\K$ is called a \emph{skeletal tree series}. 
Let $\treesrs{T}$  denote a tree series mapping trees in $\trees(\Sigma)$ to $\K$. 
By abuse of notations, given a skeletal tree $s$, we use $\treesrs{T}(s)$ for the sum of values $\treesrs{T}(t)$ for every tree $t$ of which $s=\skel(t)$.
That is,
$\treesrs{T}(s)=\sum_{t\in\skelstotrees(s)}\treesrs{T}(t)$.
Thus, given a tree series $\treesrs{T}$ (possibly generated by a WCFG or an MTA) we can treat $\treesrs{T}$ as a skeletal tree series.


\section{From Positive MTAs to PCFGs}\label{sec:CMTA2PCFG}
 
Our learning algorithm for probabilistic grammars builds on the relation between WCFGs with positive weights (henceforth PWCFGs) and PCFGs~\cite{abney1999relating,smith2007weighted}. In particular, we first establish that a \emph{positive multiplicity tree automaton} (PMTA), which is a multiplicity tree automaton (MTA) where all weights of both $\mu$ and $\lambda$ are positive, can be transformed into an equivalent  WCFG $\grmr{W}$.
That is, we  show that  a given PMTA $\aut{A}$  over a skeletal alphabet can be converted into a WCFG $\grmr{W}$ such that for every structured string $s$ we have that $\aut{A}(s)=\aut{W}(s)$. 
If the PMTA defines a convergent tree series (namely the sum of weights of all trees is finite) then so will the constructed WCFG. 
Therefore, given that the WCFG describes a probability distribution, we can apply the transformation of WCFG to a PCFG~\cite{abney1999relating,smith2007weighted} to yield a PCFG $\grmr{P}$ such that $\grmr{W}(s)=\grmr{P}(s)$, obtaining the desired PCFG for the unknown tree series.

\subsection{Transforming a PMTA into a PWCFG}
Let $\aut{A}=(\Sigma,\mathbb{R}_+,d,\mu,\lambda)$ be a PMTA over the skeletal alphabet $\Sigma=\{\Sigma_0,\Sigma_1,\ldots,\Sigma_p\}$. We define a PWCFG $\grmr{W}_\aut{A}=(\grmr{G}_\aut{A},\theta)$ for $\grmr{G}_{\aut{A}}=(\Vars,\Sigma_{0},R,S)$ as provided in Fig.~\ref{fig:eqs-pmta-to-pcfg} where $c^i_{i_{1},i_{2},...,i_{k}}$ is the respective coefficient in the matrix corresponding to $\mu_?$ for $?\in\Sigma_k$, ${1\leq k\leq p}$.

\begin{figure}[h]
	\centering
	\scalebox{1}{	
		{\hspace*{-5pt}
					$\begin{array}{l@{\,}l@{\quad}l}
					\Vars=& \{S\}\cup\{V_{i}~|~1\leq i\leq d\} &    \\
					R =&  \{  S\rightarrow V_{i} ~|~1\leq i\leq d\}\ \cup &   \phantom{..} \theta( S\rightarrow V_{i}) = \lambda[i] \\
					& \{  V_{i}\rightarrow \sigma  ~|~1\leq i\leq d,\ \sigma\in\Sigma_0 \}\  \cup  &  \phantom{..} \theta( V_{i}\rightarrow \sigma ) = \mu_\sigma[i] \\
					& \left\{ V_{i}\rightarrow V_{i_{1}}V_{i_{2}}...V_{i_{k}} ~\left|~ 1\leq i,i_1,\ldots,i_k \leq d, \  1 \leq k \leq p \right.\right\}  & 
					\begin{array}{l} \theta(V_{i}\rightarrow V_{i_{1}}V_{i_{2}}...V_{i_{k}}) =   c^i_{i_{1},i_{2},...,i_{k}}  \end{array}
					\end{array}$}}		
	\caption{Transforming a PMTA into a PCFG}\label{fig:eqs-pmta-to-pcfg}
\end{figure}
\begin{exa}
Let $\mathcal{A}=\langle \Sigma,\mathbb{R}_{+},2,\mu,\lambda\rangle$ where $\Sigma=\{\Sigma_{0}=\{a,b\}\cup\Sigma_{2}=\{?\}\}$,
and
\begin{center}
				$\lambda=\begin{pmatrix}1\\0
				\end{pmatrix}$ \ \  
				$\mu_{a}= \begin{pmatrix}0.25\\0.25
				\end{pmatrix}$ \ \  
				$\mu_{b}= \begin{pmatrix}0.25\\0.0
				\end{pmatrix}$ \ \  
				$\mu_{?} = \begin{pmatrix}
				0.25 & 0.25 & 0 & 0 \\[1mm]
				0 & 0 & 0 & 0.75
				\end{pmatrix}$.
\end{center}				
\medskip
After applying the transformation we obtain the following PCFG:
 \begin{align*}
 S&\longrightarrow N_{1}~[1.0]   \\
 N_{1}&\longrightarrow N_{1}N_{1}~[0.25]~~~|~~ N_{1}N_{2}~[0.25]~~~|~~ a~[0.25 ]~~~|~~b~[0.25 ]  \\
 N_{2}&\longrightarrow N_{2}N_{2}~[0.75]~~~|~~a~[0.25]   \\
 \end{align*}
\end{exa}

Proposition~\ref{prop:equiv1} states that the transformation preserves the weights.
\begin{proposition}\label{prop:equiv1}
	$\grmr{W}(t)=\mathcal{A}(t)\quad$ for every $t\in \trees(\Sigma)$.
\end{proposition} 

Recall that given a WCFG $\tuple{\grmr{G},\theta}$, and a tree that can be derived from $\grmr{G}$,
namely some $t\in\deriv(\grmr{G})$, the weight of $t$ is given by $\theta(t)$.
Recall also that we are working with skeletal trees $s\in\struct(\deriv(\grmr{G}))$
and the weight of a skeletal tree $s$ is given by the sum of all derivation trees $t$
such that $s$ is the skeletal tree obtained from $t$ by replacing all non-terminals with $?$.

The following two lemmas and the following notations are used in  the proof of Proposition~\ref{prop:equiv1}.
For a skeletal tree $s$ and a non-terminal $V$ we use $\grmrwgt{V}{s}$ for the weight
of all derivation trees $t$ in which the root is labeled by non-terminal $V$ and $s$ is their skeletal form.

Assume $\grmr{G}=\langle \Vars,\Sigma,R,S\rangle$.
Lemma~\ref{lem:weight}  follows in a straightforward manner from the definition of $\mathcal{W}(\cdot)$ given in
\$\ref{subsec:pcfgs}.

\begin{lemma}\label{lem:weight}
   
    Let ${s=?(s_1,s_2,\ldots,s_k)}$. The following holds for every non-terminal $V\in\Vars$:
    %
    \[\grmrwgt{V}{s}=\sum_{(X_{1},X_{2},\ldots,X_{k})\in\Vars^k}\begin{array}{l}\theta(V\rightarrow X_{1} X_{2}\cdots X_{k})\cdot  \grmrwgt{X_{1}}{s_1}\grmrwgt{X_{2}}{s_2}\cdots\grmrwgt{X_{k}}{s_k}\end{array}\]
    %
\end{lemma}

Consider now the transformation from a PMTA to a WCFG (provided in Fig.~\ref{fig:eqs-pmta-to-pcfg}).
It associates with every dimension $i$ of the PMTA $\aut{A}=(\Sigma,\mathbb{R}_+,d,\mu,\lambda)$
a variable (i.e. non-terminal) $V_i$.
The next lemma considers the $d$-dimensional vector $\mu(s)$ computed by $\aut{A}$ and
states that its $i$-th coordinate holds the value $\grmrwgt{V_i}{s}$.

\begin{lemma}\label{lem:vec-coord-vars}
    Let $s$ be a skeletal tree, and let $\mu(s)=(v_1,v_2,\ldots,v_d)$.
    Then $v_i=\grmrwgt{V_i}{s}$ for every $1\leq i \leq d$.
\end{lemma}

\begin{proof}
The proof is by induction on the height of $s$. For the base case $h=1$. Then $s$ is a leaf, thus $s\in\Sigma$. 
Then for each $i$ we have that $v[i]=\mu(\sigma)[i]$ by definition of MTA computation.
On the other hand, by the definition of the transformation in  Fig.~\ref{fig:eqs-pmta-to-pcfg},
we have $\theta(V_i\rightarrow \sigma)=\mu(\sigma)[i]$. Thus, $\grmrwgt{V_i}{s}=\grmrwgt{V_i}{\sigma}=\mu(\sigma)[i]$, so the claim holds.

For the induction step, assume $s=?(s_{1},s_{2},...,s_{k})$. 
By the definition of a multi-linear map, for each $i$ we have:
\begin{equation*}
        v_i=\sum_{\left\{(j_1,j_2,\ldots,j_k)\in \{1,2,\ldots,d\}^k\right\}} c^i_{j_{1},...,j_{k}}v_{i}[j_{1}]\cdot...\cdot v_{k}[j_{k}]
\end{equation*}
where $c^i_{j_{1},...,j_{k}}$ are the coefficients of the $d\times d^k$ matrix of $\mu_?$ for $?\in\Sigma_k$.
By the definition of the transformation in  Fig.~\ref{fig:eqs-pmta-to-pcfg} we have that $c^i_{j_{1},...,j_{k}}=\theta(V_{i}\rightarrow V_{j_{1}}V_{j_{2}}...V_{j_{k}})$. 
Also, from our induction hypothesis, we have that for each $j_{i}$, $v_{i}[j_{i}]=\grmrwgt{V_{j_i}}{s_i}$. Therefore, we have that:
\[
        v_i=\sum_{V_{j_{1}}V_{j_{2}}...V_{j_{p}}\in \Vars^{k}} \begin{array}{l}
            \theta(V_{i}\rightarrow V_{j_{1}}V_{j_{2}}...V_{j_{k}})\cdot
            \grmrwgt{V_{j_1}}{s_1}\cdot...\cdot \grmrwgt{V_{j_k}}{s_k}
            \end{array}
\]
which according to Lemma \ref{lem:weight} is equal to $\grmrwgt{V_i}{s}$ as required.
\end{proof}

We are now ready to prove {Proposition~\ref{prop:equiv1}}. 
\commentout{
 which states that\\
\begin{itemize}
    \item []
    \emph{  
    $\grmr{W}(t)=\mathcal{A}(t)\quad $ for every $t\in \skel(\trees(\Sigma))$. }
\end{itemize}}

\begin{proof}[Proof of Prop.\ref{prop:equiv1}]
 
Let $\mu(t)=v=(v_{1},v_{2},...,v_{n})$ be the vector calculated by $\mathcal{A}$ for $t$. 
The value calculated by $\mathcal{A}$ is $\lambda\cdot v$, which is:
\begin{equation*}
    \sum_{i=1}^{n} v_{i}\cdot \lambda[i]
\end{equation*}
By the transformation in  Fig.~\ref{fig:eqs-pmta-to-pcfg} we have that $\lambda[i]=\theta(S\rightarrow V_{i})$ for each $i$. So we have:
\begin{equation*}
    \sum_{i=1}^{n} v_{i}\cdot \lambda[i] = \theta(S\rightarrow V_{i})\cdot v_{i}
\end{equation*}
By Lemma~\ref{lem:vec-coord-vars}, for each $i$, $v_{i}$ is equal to the probability of deriving $t$ starting from the non-terminal $V_{i}$, so we have that the value calculated by $\mathcal{A}$ is the probability of deriving the tree starting from the start symbol $S$. That is,  
$\grmr{W}(t)=\grmrwgt{S}{t}=\mathcal{A}(t)$.
\end{proof}

\subsection{Structurally Unambiguous WCFGs and PCFGs}

In this paper we consider \emph{structurally unambiguous} WCFGs and PCFGs (in short SUWCFGs and SUPCFGs, resp.) as defined in the sequel in \S\ref{sec:sucfg}.
In  Thm.~\ref{thm:bounds}, given in \S\ref{sec:learning-CMTAs}, we show that we can learn a PMTA for a SUWCFG in polynomial time using a polynomial number of queries (see exact bounds there), thus obtaining the following result.
 \begin{corollary}\label{cor:learnSUWCFG}
	SUWCFGs can be learned in polynomial time using \smq s and 
	\seq s, where the number of \seq s is bounded by the number of non-terminal symbols.
\end{corollary}

The overall learning time for SUPCFG relies, on top of Corollary~\ref{cor:learnSUWCFG}, on the complexity of 
converting a WCFG into a PCFG~\cite{abney1999relating}, for which an exact bound is not provided, but the method is reported to converge quickly~\cite[\S 2.1]{smith2007weighted}.
%

\section{Learning of Structurally Unambiguous Probabilistic Grammars}

In this section we discuss the setting of the algorithm, the ideas behind Angluin-style learning algorithms, and the issues with using current algorithms to learn PCFGs.
As in \gstar (the algorithm for CFGs~\cite{Sakakibara88}), we assume an oracle that can answer two types of queries: \emph{structured membership queries} (\smq) and \emph{structured equivalence queries} (\seq) regarding the unknown regular tree series $\treesrs{T}$ (over a given ranked alphabet $\Sigma$). Given a structured string $s$, the query $\smq(s)$ is answered with the value $\treesrs{T}(s)$. Given an automaton $\aut{A}$ the query $\seq(\aut{A})$  is answered ``yes'' if $\aut{A}$ implements the skeletal tree series $\treesrs{T}$ and otherwise the answer is a pair $(s,\treesrs{T}(s))$ where $s$ is a  structured string for which $\treesrs{T}(s)\neq \aut{A}(s)$ (up to a  predefined error). 


Our starting point is the learning algorithm \mstar~\cite{habrard2006learning} which learns MTA using \smq s and \seq s. 
 First we explain the idea behind this and similar algorithms, next the issues with applying it as is for learning PCFGs, then the idea behind restricting attention to strucutrally unambiguous grammars, and finally our algorithm itself.

\paragraph{Hankel Matrix}
The \emph{Hankel Matrix} is a key concept in learning of languages and formal series.
A word or tree language as well as a word or tree series can be represented by its Hankel Matrix. 
The Hankel Matrix has infinitely many rows and infinitely many columns. In the case of word series both rows and columns correspond to an infinite enumeration $w_0,w_1,w_2,\ldots$ of words over the given alphabet. In the case of tree series, the rows correspond to an infinite enumeration of trees $t_0,t_1,t_2,\ldots$ (where $t_i\in\trees(\Sigma)$) and the columns to an infinite enumeration of contexts $c_0,c_1,c_2,\ldots$ (where $c_i\in\trees_{\context}(\Sigma)$). In the case of words, the entry $\hm{H}(i,j)$ holds the value for the word $w_i\cdot w_j$, and in the case of trees it holds the value of the tree $\contextConcat{c_j}{t_i}$.
If the series is \emph{regular} there should exists a finite number of rows in this infinite matrix, which we term \emph{basis}, such that all other rows can be represented using rows in the basis. In the case of \lstar, and \gstar (that learn word-languages and tree-languages, resp.) rows that are not in the basis should be equal to rows in the basis. The rows of the basis correspond to the automaton states, and the equalities to other rows determines the transition relation. In the case of \mstar (that learns tree-series by means of multiplicity tree automata) and its predecessor (that learns word-series using multiplicity word automata) rows not in the basis should be expressible as a linear combination of rows in the basis, and the linear combinations determines the weights of the extracted automaton. 

\commentout{
\begin{figure}
    \centering
    \scalebox{.8}{
        \begin{tikzpicture}
        \node[
            draw=none] at (-0.5,0) {    \begin{tabular}{ c|c|c|c|c|c|c|c|c|c|c|c } 
      & $\varepsilon$ & $a$ & $b$ & $aa$ & $ab$ & $ba$ & $bb$ & $aaa$ & $aab$ & $aba$ & $\hdots$ \\ \hline
      $\varepsilon$ & $0$ & $1$ & $0$ & $2$ & $1$ & $1$ & $0$ & $3$ & $2$ & $2$ & $\ddots$ \\ \hline
      $a$ & $1$ & $2$ & $1$ & $3$ & $2$ & $2$ & $1$ & $4$ & $3$ & $3$ & $\ddots$ \\ \hline
      $b$ & $0$ & $1$ & $0$ & $2$ & $1$ & $1$ & $0$ & $3$ & $2$ & $2$ & $\ddots$ \\ \hline
      $aa$ & $2$ & $3$ & $2$ & $4$ & $3$ & $3$ & $2$ & $5$ & $4$ & $4$ & $\ddots$ \\ \hline
     $ab$ & $1$ & $2$ & $1$ & $3$ & $2$ & $2$ & $1$ & $4$ & $3$ & $3$ & $\ddots$ \\ \hline
     $ba$ & $1$ & $2$ & $1$ & $3$ & $2$ & $2$ & $1$ & $4$ & $3$ & $3$ & $\ddots$ \\ \hline
     $bb$ & $0$ & $1$ & $0$ & $2$ & $1$ & $1$ & $0$ & $3$ & $2$ & $2$ & $\ddots$ \\ \hline
     $aaa$ & $3$ & $4$ & $3$ & $5$ & $4$ & $4$ & $3$ & $6$ & $5$ & $6$ & $\ddots$ \\ \hline
     $\vdots$ & $\ddots$ & $\ddots$ & $\ddots$ & $\ddots$ & $\ddots$ & $\ddots$ & $\ddots$ & $\ddots$ & $\ddots$ & $\ddots$ & $\ddots$\\  
    \end{tabular}};
            \node[draw=black,text=black,fill=gray!30,text width=0.1\textwidth,
            align=center,anchor=center] at (7,2.5) {$v_{1}$};
            \node[draw=black,text=black,fill=gray!30,text width=0.1\textwidth,
            align=center, anchor=center] at (7,1.8) {$v_{2}$};
            \node[draw=black,text=black,fill=gray!30,text width=0.1\textwidth,
            align=center, anchor=center] at (7,1.1) {$v_{1}$};
            \node[draw=black,text=black,fill=gray!30,text width=0.15\textwidth,
            align=center, anchor=center] at (7,0.4) {$2\cdot v_{2}-v_{1}$};
            \node[draw=black,text=black,fill=gray!30,text width=0.1\textwidth,
            align=center, anchor=center] at (7,-0.3) {$v_{1}$};
            \node[draw=black,text=black,fill=gray!30,text width=0.1\textwidth,
            align=center, anchor=center] at (7,-1) {$v_{1}$};
            \node[draw=black,text=black,fill=gray!30,text width=0.1\textwidth,
            align=center, anchor=center] at (7,-1.7) {$v_{2}$};
            \node[draw=black,text=black,fill=gray!30,text width=0.15\textwidth,
            align=center, anchor=center] at (7,-2.4) {$3\cdot v_{2}-2\cdot v_{1}$};
    \end{tikzpicture}}
    \caption{The Hankel Matrix for the word-series that returns the number of $a$'s in a word over the alphabet $\Sigma=\{a,b\}$.
    }\label{fig:hankelWordSeries}
\end{figure}

Fig.~\ref{fig:hankelWordSeries} depicts the Hankel Matrix for the word-series that returns the number of $a$'s in a word over the alphabet $\Sigma=\{a,b\}$. The first two rows linearly span the entire matrix. Let $H[\varepsilon]=v_{1}$ and $H[a]=v_{2}$. The rest of the rows can be described as a linear combination of $v_{1}$ and $v_{2}$, as indicated in gray one the right.
\end{exa}}

\begin{figure}
    \centering
    \scalebox{.9}{
    \begin{tikzpicture}
        \node[draw=none] at (0,0) {\begin{tabular}{ c|c|c|c|c|c } 
  & $\context$ & \begin{tikzpicture}
 \node[shape=circle,draw=none] (root) at (0.5,1) {\tiny{$b$}};
 \node[shape=circle,draw=none] (A) at (0,0) {\tiny{$\context$}};
 \node[shape=circle,draw=none] (A2) at (1,0) {\tiny{$a$}};
 \path [-] (root) edge node[above] {} (A);
 \path [-] (root) edge node[above] {} (A2);
 \end{tikzpicture} & \begin{tikzpicture}
 \node[shape=circle,draw=none] (root) at (0.5,1) {\tiny{$b$}};
 \node[shape=circle,draw=none] (A) at (0,0) {\tiny{$a$}};
 \node[shape=circle,draw=none] (A2) at (1,0) {\tiny{$\context$}};
 \path [-] (root) edge node[above] {} (A);
 \path [-] (root) edge node[above] {} (A2);
 \end{tikzpicture} & \begin{tikzpicture}
 \node[shape=circle,draw=none] (root) at (0.5,2) {\tiny{$b$}};
 \node[shape=circle,draw=none] (A1) at (0,1) {\tiny{$\context$}};
 \node[shape=circle,draw=none] (B1) at (1,1) {\tiny{$b$}};
 \node[shape=circle,draw=none] (A2) at (1.5,0) {\tiny{$a$}};
 \node[shape=circle,draw=none] (A3) at (0.5,0) {\tiny{$a$}};
 \path [-] (root) edge node[above] {} (A1);
 \path [-] (root) edge node[above] {} (B1);
 \path [-] (B1) edge node[above] {} (A2);
 \path [-] (B1) edge node[above] {} (A3);
 \end{tikzpicture}&$\hdots$\\ \hline
 
 $a$ & $1$ & $2$ & $2$ & $3$ &$\ddots$\\ \hline
 
 \begin{tikzpicture}
 \node[shape=circle,draw=none] (root) at (0.5,1) {\tiny{$b$}};
 \node[shape=circle,draw=none] (A) at (0,0) {\tiny{$a$}};
 \node[shape=circle,draw=none] (A2) at (1,0) {\tiny{$a$}};
 \path [-] (root) edge node[above] {} (A);
 \path [-] (root) edge node[above] {} (A2);
 \end{tikzpicture} & $2$ & $3$ & $3$ & $4$ &$\ddots$ \\  \hline
 \begin{tikzpicture}
 \node[shape=circle,draw=none] (root) at (0.5,2) {\tiny{$b$}};
 \node[shape=circle,draw=none] (A1) at (0,1) {\tiny{$a$}};
 \node[shape=circle,draw=none] (B) at (1,1) {\tiny{$b$}};
 \node[shape=circle,draw=none] (A2) at (0.5,0) {\tiny{$a$}};
 \node[shape=circle,draw=none] (A3) at (1.5,0) {\tiny{$a$}};
 \path [-] (root) edge node[above] {} (A1);
 \path [-] (root) edge node[above] {} (B); 
 \path [-] (B) edge node[above] {} (A2);
 \path [-] (B) edge node[above] {} (A3);
 \end{tikzpicture} & $3$ & $4$ & $4$ & $5$ &$\ddots$\\  \hline
 
 \begin{tikzpicture}
 \node[shape=circle,draw=none] (root) at (0.5,2) {\tiny{$b$}};
 \node[shape=circle,draw=none] (A1) at (1,1) {\tiny{$a$}};
 \node[shape=circle,draw=none] (B) at (0,1) {\tiny{$b$}};
 \node[shape=circle,draw=none] (A2) at (-0.5,0) {\tiny{$a$}};
 \node[shape=circle,draw=none] (A3) at (0.5,0) {\tiny{$a$}};
 \path [-] (root) edge node[above] {} (A1);
 \path [-] (root) edge node[above] {} (B); 
 \path [-] (B) edge node[above] {} (A2);
 \path [-] (B) edge node[above] {} (A3);
 \end{tikzpicture} & $3$ & $4$ & $4$ & $5$ &$\ddots$ \\  \hline
 
 \begin{tikzpicture}
 \node[shape=circle,draw=none] (root) at (0,2) {\tiny{$b$}};
 \node[shape=circle,draw=none] (B1) at (1,1) {\tiny{$b$}};
 \node[shape=circle,draw=none] (B2) at (-1,1) {\tiny{$b$}};
 \node[shape=circle,draw=none] (A1) at (1.5,0) {\tiny{$a$}};
 \node[shape=circle,draw=none] (A2) at (0.5,0) {\tiny{$a$}};
 \node[shape=circle,draw=none] (A3) at (-1.5,0) {\tiny{$a$}};
 \node[shape=circle,draw=none] (A4) at (-0.5,0) {\tiny{$a$}};
 \path [-] (root) edge node[above] {} (B1);
 \path [-] (root) edge node[above] {} (B2); 
 \path [-] (B1) edge node[above] {} (A1);
 \path [-] (B1) edge node[above] {} (A2);
 \path [-] (B2) edge node[above] {} (A3);
 \path [-] (B2) edge node[above] {} (A4);
 \end{tikzpicture} & $4$ & $5$ & $5$ & $6$ &$\ddots$ \\ \hline
  $\vdots$ & $\ddots$ & $\ddots$ & $\ddots$ & $\ddots$ &$\ddots$\\ 
\end{tabular}};
\node[draw=black,text=black,fill=gray!30,text width=0.1\textwidth,
            align=center,anchor=center] at (8,3.9) {$v_{1}$};
            \node[draw=black,text=black,fill=gray!30,text width=0.1\textwidth,
            align=center, anchor=center] at (8,2.7) {$v_{2}$};
            \node[draw=black,text=black,fill=gray!30,text width=0.15\textwidth,
            align=center, anchor=center] at (8,0.7) {$2\cdot v_{2}-v_{1}$};
            \node[draw=black,text=black,fill=gray!30,text width=0.15\textwidth,
            align=center, anchor=center] at (8,-2.2) {$2\cdot v_{2}-v_{1}$};
            \node[draw=black,text=black,fill=gray!30,text width=0.15\textwidth,
            align=center, anchor=center] at (8,-5) {$3\cdot v_{2}-2\cdot v_{1}$};
    \end{tikzpicture}
    }
    \caption{The Hankel Matrix for the tree-series that returns the number of leaves in a tree over the ranked alphabet $\Sigma=\{\Sigma_0,\Sigma_2\}$ where $\Sigma_{0}=\{a\}$ and $\Sigma_{2}=\{b\}$.}
    \label{fig:hankelTreeSeries}
\end{figure}

\begin{exa}\label{exa:extract-mta}
Fig.~\ref{fig:hankelTreeSeries} depicts the Hankel Matrix for the tree-series that returns the number of leaves in a tree over the ranked alphabet $\Sigma=\{\Sigma_{0}=\{a\},\Sigma_{2}=\{b\}\}$. The first two rows linearly span the entire matrix. Let $H[a]=v_{1}$ and $H[b(a,a)]=v_{2}$. The rest of the rows can be described as a linear combination of $v_{1}$ and $v_{2}$, as shown on the right in gray.
\end{exa}

\paragraph{\lstar-style query learning algorithms}
All the generalizations of \lstar\ 
share a general idea that can be explained as follows. The algorithm maintains an \emph{observation table}, a finite sub-matrix of the Hankel Matrix, whose entries
are filled by asking membership queries. Once the table meets certain criteria, namely is \emph{closed} with respect to a \emph{basis}, which is a subset of the rows,
an automaton can be extracted from it. 
In the case of \lstar\ (that learns regular word-languages) a table is closed if the row of the empty string is in the basis, for every row $s$ in the basis, its one letter extension, $s\sigma$ is in the rows of the table, and every row in the table is equivalent to some row in the basis.
In the case of \cstar (that learns tree automata) instead of one letter extensions, we need for every letter $\sigma\in\Sigma_k$, a row
with root $\sigma$ and for the $k$-children all options of trees from the basis. The criterion for a row not in the basis to be covered, is the same; it should be equivalent to a row in the basis.
In the case of \mstar, the extension requirement is as for \cstar, since we are dealing with trees. As for covering the
rows of the table that are not in the basis, the requirement is that 
a rows that is not in the basis be expressible as a linear combination of rows in the basis.

In the case of \lstar, the extracted automaton will have one state for every row in the basis, where $\epsilon$ is the initial state, and a state $s$ is determined to be final
if the observation table entry corresponding to row $s$ and column $\epsilon$ is labeled $1$ (i.e. the word $s$ is in the language). Then, the transition from $s$ on $\sigma$
is determined to be the row $s'$ of the basis that is equivalent to $s\sigma$.

In the case of $\mstar$ the size of the basis determines the dimension of the extracted multiplicity tree automaton.
Suppose the size of the basis is $d$.
The output vector $\lambda$ is set to $(c_1,\ldots,c_d)$ where $c_i$ is the value of the entry corresponding to the row $i$ of the basis and the column $\context$ (i.e the value of the $i$-th row of the basis).
Consider a letter $\sigma\in\Sigma_k$, its corresponding transition matrix $\mu_\sigma$ is a $d\times d^k$ matrix. Think of $\mu_\sigma$ as $d^k$ $d$-vectors $\left(v_{11\ldots1}\ v_{11\ldots 2}\ \ldots\ v_{dd\ldots d}\right)$; 
then the vector $v_{j_1j_2\ldots j_k}$ which corresponds to the row extending $\sigma$ with the trees in the base rows $r_{j_1},r_{j_2}\ldots r_{j_k}$ 
gets its coefficient from the linear combination for this row, see Example~\ref{exa:extract-mta}.

 Note that in all cases we would like the basis to be \emph{minimal} in the sense that no row in the basis is expressible using other rows in the basis. This is since the size of the basis derives the dimension  of the automaton, and obviously we prefer smaller automata.

\begin{figure}
\scalebox{.9}{
    \begin{minipage}{0.75\textwidth}
        \begin{tabular}{ c|c|c|c|c } 
  && $\context$ & 
  \begin{tikzpicture}
     \node[shape=circle,draw=none] (root) at (0.5,1) {\tiny{$b$}};
     \node[shape=circle,draw=none] (A) at (0,0) {\tiny{$\context$}};
     \node[shape=circle,draw=none] (A2) at (1,0) {\tiny{$a$}};
     \path [-] (root) edge node[above] {} (A);
     \path [-] (root) edge node[above] {} (A2);
     \end{tikzpicture} & \begin{tikzpicture}
     \node[shape=circle,draw=none] (root) at (0.5,1) {\tiny{$b$}};
     \node[shape=circle,draw=none] (A) at (0,0) {\tiny{$a$}};
     \node[shape=circle,draw=none] (A2) at (1,0) {\tiny{$\context$}};
     \path [-] (root) edge node[above] {} (A);
     \path [-] (root) edge node[above] {} (A2);
 \end{tikzpicture} \\ \hline
 
 \cellcolor{red!50} $v_{1}$&$a$ & \cellcolor{purple!50} $1$ & $2$ & $2$ \\ \hline
 
 \cellcolor{orange!50} $v_{2}$&
     \begin{tikzpicture}
     \node[shape=circle,draw=none] (root) at (0.5,1) {\tiny{$b$}};
     \node[shape=circle,draw=none] (A) at (0,0) {\tiny{$a$}};
     \node[shape=circle,draw=none] (A2) at (1,0) {\tiny{$a$}};
     \path [-] (root) edge node[above] {} (A);
     \path [-] (root) edge node[above] {} (A2);
     \end{tikzpicture} & \cellcolor{purple!50} $2$ & $3$ & $3$ \\  \hline
     \cellcolor{yellow!50} $2\cdot v_{2}-v_{1}$&\begin{tikzpicture}
     \node[shape=circle,draw=none] (root) at (0.5,2) {\tiny{$b$}};
     \node[shape=circle,draw=none] (A1) at (0,1) {\tiny{$a$}};
     \node[shape=circle,draw=none] (B) at (1,1) {\tiny{$b$}};
     \node[shape=circle,draw=none] (A2) at (0.5,0) {\tiny{$a$}};
     \node[shape=circle,draw=none] (A3) at (1.5,0) {\tiny{$a$}};
     \path [-] (root) edge node[above] {} (A1);
     \path [-] (root) edge node[above] {} (B); 
     \path [-] (B) edge node[above] {} (A2);
     \path [-] (B) edge node[above] {} (A3);
     \end{tikzpicture} & $3$ & $4$ & $4$ \\  \hline
 \cellcolor{green!50} $2\cdot v_{2}-v_{1}$&
 \begin{tikzpicture}
     \node[shape=circle,draw=none] (root) at (0.5,2) {\tiny{$b$}};
     \node[shape=circle,draw=none] (A1) at (1,1) {\tiny{$a$}};
     \node[shape=circle,draw=none] (B) at (0,1) {\tiny{$b$}};
     \node[shape=circle,draw=none] (A2) at (-0.5,0) {\tiny{$a$}};
     \node[shape=circle,draw=none] (A3) at (0.5,0) {\tiny{$a$}};
     \path [-] (root) edge node[above] {} (A1);
     \path [-] (root) edge node[above] {} (B); 
     \path [-] (B) edge node[above] {} (A2);
     \path [-] (B) edge node[above] {} (A3);
 \end{tikzpicture} & $3$ & $4$ & $4$ \\  \hline
 \cellcolor{blue!40} $3\cdot v_{2}-2\cdot v_{1}$&
 \begin{tikzpicture}
     \node[shape=circle,draw=none] (root) at (0,2) {\tiny{$b$}};
     \node[shape=circle,draw=none] (B1) at (1,1) {\tiny{$b$}};
     \node[shape=circle,draw=none] (B2) at (-1,1) {\tiny{$b$}};
     \node[shape=circle,draw=none] (A1) at (1.5,0) {\tiny{$a$}};
     \node[shape=circle,draw=none] (A2) at (0.5,0) {\tiny{$a$}};
     \node[shape=circle,draw=none] (A3) at (-1.5,0) {\tiny{$a$}};
     \node[shape=circle,draw=none] (A4) at (-0.5,0) {\tiny{$a$}};
     \path [-] (root) edge node[above] {} (B1);
     \path [-] (root) edge node[above] {} (B2); 
     \path [-] (B1) edge node[above] {} (A1);
     \path [-] (B1) edge node[above] {} (A2);
     \path [-] (B2) edge node[above] {} (A3);
     \path [-] (B2) edge node[above] {} (A4);
\end{tikzpicture} & $4$ & $5$ & $5$\\ 
\end{tabular}
    \end{minipage}%
    \begin{minipage}{0.15\textwidth}
    \begin{align*}
        \mu_{a}=&\begin{bmatrix}&\cellcolor{red!50}1&\\&\cellcolor{red!50}0&\end{bmatrix}\\
        \mu_{b}=&\begin{bmatrix}&\cellcolor{orange!50}0&\cellcolor{yellow!50}-1&\cellcolor{green!50}-1&\cellcolor{blue!40}-2&\\&\cellcolor{orange!50}1&\cellcolor{yellow!50}2&\cellcolor{green!50}2&\cellcolor{blue!40}3&\end{bmatrix}\\
        \lambda=&\begin{bmatrix}&\cellcolor{purple!50}1&\\&\cellcolor{purple!50}2&\end{bmatrix}
    \end{align*}
    \end{minipage}%
    }
    \caption{A closed observation table (on the left) and the MTA extracted from it (on the right).}
    \label{fig:mta_extract}
\end{figure}

\begin{exa}
\label{exa:extract-mta}
Fig.~\ref{fig:mta_extract} shows (on the left) a closed observation table which is a sub-matrix of the Hankel Matrix of Fig.~\ref{fig:hankelTreeSeries},
and the  MTA extracted from it (on the right). The base here consists of the first two rows (thus the dimension $d$ is $2$) and as can be seen they linearly span the other rows of the table. 
The coefficients with which they span the table appear on the left-most column. 
The final vector $\lambda$ is set to $\biword{1}{2}$ since these are the values of the row basis in column $\context$.
The matrix $\mu_a$ is a $2\times 2^0$ matrix since $a\in\Sigma_0$. We view it as a single vector that corresponds to the row of the letter $a$, it is thus $\biword{1}{0}$. The matrix $\mu_b$ is a $2\times 2^2$ matrix since $b\in\Sigma_2$. We view it as the vectors $(v_{11}\ v_{12}\ v_{21}\ v_{22})$.
Consider for instance $v_{21}$. It corresponds to the tree with root $b$ whose left child is the second base row, namely the tree on the second row of the observation table, and whose right child is the first base row, namely $a$. It thus corresponds to the 4th row of the table and is therefore 
set to $\biword{-1}{2}$. The other vectors of $\mu_b$ are set similarly (follow the color coding in Fig.~\ref{fig:mta_extract}).   
\end{exa}

Back to the problem at hand, since we are interested in probabilistic languages, we would like to apply an \lstar-style algorithm
to obtain a PMTA, so that we can convert it into a PCFG. Since the coefficients of the linear combination in rows that 
are not in the basis determine the values of the entries of the transitions matrix, we cannot have negative coefficients.
That is, we need the algorithm to find a subset of the rows (a basis), so that all other rows can be obtained by a \emph{positive linear combination} of the rows in the  basis, so that the extracted automaton will be a PMTA.

\paragraph{Positive linear spans}
An interest in \emph{positive linear combinations} occurs also in the research community studying convex cones and derivative-free optimizations and a theory of positive linear combinations has been developed~\cite{cohen1993nonnegative,Regis2016}.\footnote{Throughout the paper we use the terms \emph{positive} and \emph{nonnegative} interchangeably.} We need the following definitions and results.

The \emph{positive span} of a finite set of vectors $S=\{v_{1},v_{2},...,v_{k}\}\subseteq\mathbb{R}^{n}$ is defined as follows:
\begin{equation*}
\pos(S)=\{\lambda_{1}\cdot v_{1}+\lambda_{2}\cdot v_{2}+...+\lambda_{k}\cdot v_{k} ~|~ \lambda_{i}\geq 0, \forall 1\leq i\leq k\}
\end{equation*}
A set of vectors $S=\{v_{1},v_{2},...,v_{k}\}\subseteq\mathbb{R}^{n}$ is \emph{positively dependent} if some $v_{i}$ is a positive combination of the other vectors; otherwise, it is \emph{positively independent}. 
Let $A\in\mathbb{R}^{m\times n}$. We say that $A$ is \emph{nonnegative} if all of its elements are nonnegative.  The \emph{nonnegative column (resp. row) rank}  of $A$, denoted $\crank(A)$ (resp. $\rrank(A)$), is defined as the smallest nonnegative integer $q$ for which there exist a set of column- (resp. row-) vectors $V=\{v_{1},v_{2},...,v_{q}\}$  in $\mathbb{R}^{m}$ such that every column (resp. row) of $A$ can be represented as a positive combination of $V$. 
It was shown  that $\crank(A)=\rrank(A)$ for any matrix $A$~\cite{cohen1993nonnegative}. Thus one can  freely use $\prank(A)$ for \emph{positive rank}, to refer to either one of these.

\subsection{Trying to Accommodate M* for learning PMTA}\label{sec:mstarpmta}
The first question that comes to mind, is whether we can use the \mstar\ algorithm as is in order to learn a positive tree series. We show that this is not the case.
We prove two propositions, which show that one cannot use the \mstar algorithm to learn PMTA. Proposition \ref{prop:mstar_negative_mta} shows a grammar, for which there exist a PMTA, and a positive base in the Hankel Matrix, but applying the \mstar algorithm as is may choose a non-negative base.

\begin{prop}\label{prop:mstar_negative_mta}
    There exists a probabilistic word series for which the \mstar\ algorithm may return an MTA with negative weights.
\end{prop}

\begin{proof}
    Let $\mathcal{G}$ be the grammar that assigns to the words $aa$,$ab$,$ac$,$ba$,$cb$,$cc$ a probability of $\frac{1}{6}$ and to all other words a probability of $0$. For simplicity we assume that the grammar is regular.
	The first rows of Hankel Matrix for this word series are given in  Fig.\ref{fig:mta_non_positive} (all entries not in the figure are $0$).
	One can see that the rows $\epsilon$, $b$, $c$, $ba$ (highlighted in light gray) are a positive span of the entire Hankel Matrix. 
	However, the $\mstar$ algorithm may return the MTA spanned by the basis $\epsilon$, $a$, $b$, $aa$ (highlighted in dark gray). 
	Since the row of $c$ is obtained by subtracting the row of $b$ from the row of $a$, the resulting MTA will contain negative weights.
\end{proof}
	
\begin{figure}
    \centering
        \begin{tabular}{ c|c|c|c|c|c|c|c|c|c|c|c|c|c } 
  & $\varepsilon$ & $a$ & $b$ & $c$& $aa$ & $ab$ &$ac$& $ba$& $bb$& $bc$& $ca$ &$cb$& $cc$\\ \hline
  
 \diagfil{2cm}{gray!50}{gray!13}{$\varepsilon$}& $0$ & $0$ & $0$ & $0$&$\frac{1}{6}$&$\frac{1}{6}$&$\frac{1}{6}$&$\frac{1}{6}$&$0$&$0$&$0$&$\frac{1}{6}$&$\frac{1}{6}$ \\ \hline 
 \cellcolor{gray!50} $a$ & $0$ & $\frac{1}{6}$ & $\frac{1}{6}$ & $\frac{1}{6}$&$0$&$0$&$0$&$0$&$0$&$0$&$0$&$0$&$0$ \\ \hline
 \diagfil{2cm}{gray!50}{gray!13}{$b$} & $0$ & $\frac{1}{6}$ & $0$ & $0$&$0$&$0$&$0$&$0$&$0$&$0$&$0$&$0$&$0$\\ \hline
 \cellcolor{gray!13} $c$ & $0$ & $0$ & $\frac{1}{6}$ & $\frac{1}{6}$&$0$&$0$&$0$&$0$&$0$&$0$&$0$&$0$&$0$\\ \hline
 \cellcolor{gray!50} $aa$ & $\frac{1}{6}$ & $0$ & $0$ & $0$&$0$&$0$&$0$&$0$&$0$&$0$&$0$&$0$&$0$ \\  \hline
 $ab$ & $\frac{1}{6}$ & $0$ & $0$ & $0$&$0$&$0$&$0$&$0$&$0$&$0$&$0$&$0$&$0$ \\  \hline
 $ac$ & $\frac{1}{6}$ & $0$ & $0$ & $0$&$0$&$0$&$0$&$0$&$0$&$0$&$0$&$0$&$0$ \\  \hline
 \cellcolor{gray!13}$ba$ & $\frac{1}{6}$ & $0$ & $0$ & $0$&$0$&$0$&$0$&$0$&$0$&$0$&$0$&$0$&$0$\\ \hline
 $bb$ & $0$ & $0$ & $0$ & $0$&$0$&$0$&$0$&$0$&$0$&$0$&$0$&$0$&$0$ \\  \hline
 $bc$ & $0$ & $0$ & $0$ & $0$&$0$&$0$&$0$&$0$&$0$&$0$&$0$&$0$&$0$\\  \hline
 $ca$ & $0$ & $0$ & $0$ & $0$&$0$&$0$&$0$&$0$&$0$&$0$&$0$&$0$&$0$ \\  \hline
 $cb$ & $\frac{1}{6}$ & $0$ & $0$ & $0$&$0$&$0$&$0$&$0$&$0$&$0$&$0$&$0$&$0$ \\ \hline
 $cc$ & $\frac{1}{6}$ & $0$ & $0$ & $0$&$0$&$0$&$0$&$0$&$0$&$0$&$0$&$0$&$0$ \\
\end{tabular}
    \caption{The observation table for the grammar $\mathcal{G}$. Two of the possible bases are highlighted. The dark gray one results in a negative MTA, while the light gray one results in a positive one, that can be converted to a PCFG. For simplicity, the trees are presented as words, since we assume that the grammar is regular.}
    \label{fig:mta_non_positive}
\end{figure}

Proposition~\ref{prop:pcfg-no-finite-rank} 
 strengthens this and shows that there are grammars which have PMTA, but  
  \mstar\ would \emph{always} return a non-positive MTA.
  This is since 
the Hankel Matrix
	$H_\grmr{G}$ corresponding to the tree-series $\treesrs{T}_G$ of the respective PCFG $\grmr{G}$ has the property that
	no finite number of rows positively spans the entire matrix.
	
Before proving Proposition \ref{prop:pcfg-no-finite-rank}, we prove the following lemma:
\begin{lem}\label{lemmaSing}
	Let $B=\{b_{1},b_{2},...,b_{p}\}$ be a set of positively independent vectors. 
	Let $\hat{B}$ be a matrix whose columns are the elements of $B$, and let $\alpha$ be a positive vector. 
	Then if $\hat{B}\alpha=b_{i}$ then $\alpha[i]=1$ and $\alpha[j]=0$ for every $j\neq i$.
\end{lem}

\begin{proof}
	Assume $b_{i}=\hat{B}\alpha$. Then we have:
	\begin{equation*}
	    b_{i}=\hat{B}\alpha=\sum_{j=1}^{p}\alpha[j]b_{j}=\alpha[i]b_{i}+\sum_{j\neq i}\alpha[j]b_{j}
	\end{equation*}
	If $\alpha[i]<1$ we obtain:
	\begin{equation*}
	    b_{i}(1-\alpha[i])=\sum_{j\neq i}\alpha[j]b_{j}
	\end{equation*}
	Which is a contradiction, since $b_{i}\in B$ and thus can't be described as a positive combination of the other elements.
	
	If $\alpha[i]>1$ we obtain:
	\begin{equation*}
	    \sum_{j\neq i}\alpha[j]b_{j}+(\alpha[i]-1)b_{i}=0
	\end{equation*}
	This is a contradiction, since all $b$'s are positive, all $\alpha[j]\geq 0$ and $\alpha[i]>1$.
	
	Hence $\alpha[i]=1$. Therefore we have:
	\begin{align*}
	    \sum_{j\neq i}\alpha[j]b_{j}=0
	\end{align*}
	Since $\alpha[j]\geq 0$, and $b_{j}$ is positive the only solution is that $\alpha[j]=0$ for every $j\neq i$.
\end{proof}

\commentout{
\begin{prop}\label{prop:needPMTAs}
	There exists a probabilistic word series for which the \mstar\ alg. may return an MTA with negative weights.
\end{prop}
The proof shows this is the case for the word series over alphabet ${\Sigma=\{a,b,c\}}$ which assigns   
the following six strings: $aa$, $ab$, $ac$, $ba$, $cb$, $cc$ probability of $\frac{1}{6}$ each, and probability $0$ to all other strings.
}

We turn to ask whether 
we can adjust the algorithm \mstar\ to learn a positive basis. 
We note first that working with positive spans is much trickier than working with general spans, since for $d\geq 3$ there
is no bound on the size of a positively independent set in  $\mathbb{R}_+^d$~\cite{Regis2016}. To apply the ideas of the Angluin-style query learning algorihtms we 
 need the Hankel Matrix (which is infinite) to contain a finite sub-matrix with the same rank. 
Unfortunately, as we show next,
there exists a probabilistic (thus positive) tree series  $\treesrs{T}$
that can be recognized by a PMTA, but none of its finite-sub-matrices span the entire space of $\hm{H}_\treesrs{T}$.

\begin{prop}\label{prop:pcfg-no-finite-rank}
	There exists a PCFG $\grmr{G}$ s.t. for the Hankel Matrix
	$H_\grmr{G}$ corresponding to its tree-series $\treesrs{T}_G$ 
	no finite number of rows positively spans the entire matrix. 
\end{prop}

\commentout{
Thus, running \mstar on $\treesrs{T}_G$ would yield an MTA with negative weights.
The proof shows this is the case for the following PCFG:
$$\begin{array}{l@{\ \longrightarrow\ }l}
N_{1}& aN_{1}~[\frac{1}{2}]\ \mid\  aN_{2}~[\frac{1}{3}]\ \mid\  aa~[\frac{1}{6}]\\[2mm]
N_{2}& aN_{1}~[\frac{1}{4}]\ \mid\  aN_{2}~[\frac{1}{4}]\ \mid\  aa~[\frac{1}{2}]
\end{array}$$}

\begin{proof}
Let $\grmr{G}=(\{a\},\{N_{1},N_{2}\},R,N_{1})$ be the following PCFG:
$$\begin{array}{l@{\ \longrightarrow\ }l}
N_{1}& aN_{1}~[\frac{1}{2}]\ \mid\  aN_{2}~[\frac{1}{3}]\ \mid\  aa~[\frac{1}{6}]  \\[2mm]
N_{2}& aN_{1}~[\frac{1}{4}]\ \mid\  aN_{2}~[\frac{1}{4}]\ \mid\  aa~[\frac{1}{2}]   
\end{array}$$

We say that a tree has a \emph{chain structure} if every inner node is of branching-degree $2$ and has
one child which is a terminal. We say that a tree has a \emph{right-chain structure} (resp. \emph{left-chain structure})
if the non-terminal is always the right (resp. left) child.
Note that all trees in $\sema{\grmr{G}}$ have a right-chain structure, and the terminals are always the letter $a$.

Let us denote by $p_{n}$ the total probability of all trees with $n$ non-terminals s.t. the lowest non terminal is $N_{1}$, 
and similarly, let us denote by $q_{n}$ the total probability of all trees with $n$ non-terminals s.t. the lowest non-terminal is $N_{2}$. 

We have that $p_{0}=0$, $p_{1}=\frac{1}{6}$, and  $p_{2}=\frac{1}{12}$. We also have that $q_{0}=0$, $q_{1}=0$ and  $q_{2}=\frac{1}{6}$. 

Now, to create a tree with $n$ non-terminals, we should take a tree with $n-1$ non-terminals, ending with either $N_{1}$ or $N_{2}$, and use the last derivation. So we have:
\begin{align*}
p_{n}&=\frac{1}{2}\cdot p_{n-1}+\frac{1}{4}\cdot q_{n-1}\\
q_{n}&=\frac{1}{3}\cdot p_{n-1}+\frac{1}{4}\cdot q_{n-1}
\end{align*}

We want to express $p_{n}$ only as a function of $p_{i}$ for $i<n$, and similarly for $q_{n}$. Starting with the first equation we obtain:
\begin{align*}
p_{n}&=\frac{1}{2}\cdot p_{n-1}+\frac{1}{4}\cdot q_{n-1}\\
4\cdot p_{n+1}-2\cdot p_{n}&=q_{n}\\
4\cdot p_{n}-2\cdot p_{n-1}&=q_{n-1}
\end{align*}

And from the second equation we obtain:
\begin{align*}
q_{n}&=\frac{1}{3}\cdot p_{n-1}+\frac{1}{4}\cdot q_{n-1}\\
3\cdot q_{n+1}-\frac{3}{4}\cdot q_{n}&=p_{n}\\
3\cdot q_{n}-\frac{3}{4}\cdot q_{n-1}&=p_{n-1}
\end{align*}

Now setting these values in each of the equation, we obtain:

\begin{align*}
4\cdot p_{n+1}-2\cdot p_{n}&=\frac{1}{3}\cdot p_{n-1}+\frac{1}{4}(4\cdot p_{n}-2\cdot p_{n-1})\\
p_{n+1}&=\frac{1}{2}\cdot p_{n}+\frac{1}{12}\cdot p_{n-1}+\frac{1}{16}(4\cdot p_{n}-2\cdot p_{n-1})\\
p_{n+1}&=\frac{1}{2}\cdot p_{n}+\frac{1}{12}\cdot p_{n-1}+\frac{1}{4}\cdot p_{n}-\frac{1}{8}\cdot p_{n-1}\\
p_{n+1}&=\frac{3}{4}\cdot p_{n}-\frac{1}{24}\cdot p_{n-1}
\end{align*}

And:

\begin{align*}
3\cdot q_{n+1}-\frac{3}{4}\cdot q_{n}&=\frac{1}{2}\cdot(3\cdot q_{n}-\frac{3}{4}\cdot q_{n-1})+\frac{1}{4}\cdot q_{n-1}\\
3\cdot q_{n+1}&=\frac{9}{4}\cdot q_{n}-\frac{1}{8}\cdot q_{n-1}\\
q_{n+1}&=\frac{9}{12}\cdot q_{n}-\frac{1}{24}\cdot q_{n-1}=\frac{3}{4}\cdot q_{n}-\frac{1}{24}\cdot q_{n-1}
\end{align*}
let us denote by $t_{n}$ the probability that $\mathcal{G}$ assigns to $a^{n}$. This probability is:
\begin{equation*}
t_{n}=\frac{1}{6}\cdot p_{n-1}+\frac{1}{2}\cdot q_{n-1}
\end{equation*}
Since $$p_{n-1}=\frac{3}{4}\cdot p_{n-2}-\frac{1}{24}\cdot p_{n-3}$$ and $$q_{n-1}=\frac{3}{4}\cdot q_{n-2}-\frac{1}{24}\cdot q_{n-3}$$ we obtain:
{\small
\begin{align*}
t_{n}=&\frac{1}{6}\cdot(\frac{3}{4}\cdot p_{n-2}-\frac{1}{24}\cdot p_{n-3})+\frac{1}{2}\cdot(\frac{3}{4}\cdot q_{n-2}-\frac{1}{24}\cdot q_{n-3})\\
t_{n}=&\frac{1}{6}\cdot\frac{3}{4}\cdot p_{n-2}-\frac{1}{6}\cdot\frac{1}{24}\cdot p_{n-3}+\frac{1}{2}\cdot\frac{3}{4}\cdot q_{n-2}-\frac{1}{2}\cdot\frac{1}{24}\cdot q_{n-3}\\
t_{n}=&\frac{3}{4}\cdot(\frac{1}{6}\cdot p_{n-2}+\frac{1}{2}\cdot q_{n-2})-\frac{1}{24}\cdot(\frac{1}{6}\cdot p_{n-3}+\frac{1}{2}\cdot q_{n-3})=\\
       =&\frac{3}{4}\cdot t_{n-1}-\frac{1}{24}\cdot t_{n-2}
\end{align*}
}

Hence, overall, we obtain:
\begin{equation*}
t_{n}=\frac{3}{4}\cdot t_{n-1}-\frac{1}{24}\cdot t_{n-2}
\end{equation*}

Now, let $L$ be the skeletal-tree-language of the grammar $\mathcal{G}$, and let $H$ be the Hankel Matrix for this tree set. Note, that any tree $t$ whose structure is not a right-chain, would have $L(t)=0$, and also for every context $c$, $L(\contextConcat{c}{t})=0$. Similarly, every context $c$ who violates the right-chain structure, would have $L(\contextConcat{c}{t})=0$ for every $t$.

Let $T_{n}$ be the skeletal tree for the tree of right-chain structure, with $n$ leaves. We have that $L(T_{1})=0$, $L(T_{2})=\frac{1}{6}$, $L(T_{3})=\frac{1}{4}$, and for every $i>3$ we have $$L(T_{i})=\frac{3}{4}\cdot L(T_{i-1})-\frac{1}{24}\cdot L(T_{i-2}).$$ Let $v_{i}$ be the infinite row-vector of the Hankel matrix corresponding to $T_{i}$. We have that for every $i>3$, $$v_{i}=\frac{3}{4}\cdot v_{i-1}-\frac{1}{24}\cdot v_{i-2}.$$ Assume towards contradiction that there exists a subset of rows that is a positive base and spans the entire matrix $H$.

Let $B$ be the positive base, whose highest member (in the lexicographic order) is the lowest among all the positive bases. Let $v_{r}$ be the row vector for the highest member in this base. Thus, $v_{r+1}\in\posspan{B}$. Hence:
\begin{equation*}
v_{r+1}=\alpha \hat{B}
\end{equation*}
Also, $v_{r+1}=\frac{3}{4}\cdot v_{r}-\frac{1}{24}\cdot v_{r-1}$. Therefore,
\begin{align*}
\frac{3}{4}\cdot v_{r}-\frac{1}{24}\cdot v_{r-1}=\alpha\hat{B}\\
v_{r}=\frac{4}{3}\cdot\alpha\hat{B}+\frac{1}{18}\cdot v_{r-1}
\end{align*}

We will next show that $v_{r-1}$ and $v_r$ are co-linear, which contradicts our choice of $v_r$.
Since $v_{r-1}\in\posspan{B}$ we know $v_{r-1}=\alpha'\hat{B}$ for some $\alpha'$. Therefore, 
\begin{equation*}
v_{r}=(\frac{4}{3}\cdot\alpha+\frac{1}{18}\cdot\alpha')\hat{B}
\end{equation*}
Let $\beta=\frac{4}{3}\cdot\alpha+\frac{1}{18}\cdot\alpha'$. Since $\alpha$ and $\alpha'$ are non-negative vectors, so is $\beta$. And by Lemma \ref{lemmaSing}  it follows that for every $i\neq r$:
\begin{equation*}
\beta_{i}=\frac{4\cdot\alpha_{i}}{3}+\frac{\alpha'_{i}}{18}=0
\end{equation*}
Since $\alpha_{i}$ and $\alpha'_{i}$ are non-negative, we have that $\alpha_{i}=\alpha'_{i}=0$.

Since for every $i\neq r$, $\alpha'_{i}=0$, it follows that $v_{r-1}=m\cdot v_{r}$ for some $m\in\mathbb{R}$. 
Now, $m$ can't be zero since our language is strictly positive and all entries in the matrix are non-negative. Thus, $v_{r}=\frac{1}{m}\cdot v_{r-1}$, and $v_{r}$ and $v_{r-1}$ are co-linear. We can replace $v_{r}$ by $v_{r-1}$, contradicting the fact that we chose the base whose highest member is as low as possible.
\end{proof}

We note that a similar negative result was independently obtained in~\cite{HeerdtKR020} which studies the learnability of weighted automata, using an \lstar\ algorithm, over algebraic structures different than fields.
It was shown that weighted automata over the semiring of natural numbers are not learnable in the \lstar\ framework. The formal series used in the proof is $f(a^n)=2^n-1$. 
Since this series is divergent it does not characterize a probabilistic automata/grammars. Proposition~\ref{prop:pcfg-no-finite-rank} strengthens this result as it provides a convergent series that is not $\lstar$ learnable.


\subsection{Focusing on Structurally Unambiguous CFGs}\label{sec:sucfg}
To overcome these obstacles we restrict attention to structurally unambiguous CFGs (SUCFGs) and their weighted/probabilistic versions (SUWCFGs/SUPCFGs).
A context-free grammar is termed \emph{ambiguous} if there exists more than one derivation tree for the same word.
We  term a CFG \emph{structurally ambiguous} if there exists more than one
derivation tree with the same structure for the same word.
A context-free  language is termed \emph{inherently ambiguous}   if it cannot be derived by an unambiguous CFG.
 Note that a CFG which is unambiguous is also structurally unambiguous, while the other direction is not necessarily true.
{For instance, 
	the language $\{a^n b^n c^md^m~|~ n\geq 1, m\geq 1\}\cup \{a^n b^m c^md^n~|~ n\geq 1,m\geq 1\}$ which is inherently ambiguous~\cite[Thm.~4.7]{HopcroftUllman79} is not inherently structurally ambiguous.} 
Therefore
we have relaxed the classical unambiguity requirement.

\paragraph{The Hankel Matrix and MTA for SUPCFG}
Recall that the Hankel Matrix considers skeletal trees. Therefore if a word has more than one derivation tree with the same
structure, the respective entry in the matrix holds the sum of weights for all derivations. This makes it harder for the learning algorithm
to infer the weight of each tree separately. By choosing to work with structurally unambiguous grammars, we overcome this difficulty
as an entry corresponds to a single derivation tree. 

To discuss properties of the Hankel Matrix for an SUPCFG we need the following definitions.
Let ${H}$ be a matrix, $t$ a tree (or row index) $c$ a context (or column index), $T$ a set of trees (or row indices) and $C$ a set of contexts (or column indices).
We use ${H}[t]$ (resp. ${H}[c]$) for the row (resp. column) of ${H}$ corresponding to $t$ (resp. $c$).
Similarly we use ${H}[T]$ and ${H}[C]$ for the corresponding sets of rows or columns. Finally, we use  ${H}[t][C]$ for the restriction of ${H}$ to row $t$ and columns $[C]$.

Two vectors, $v_{1},v_{2}\in\mathbb{R}^{n}$ are \emph{co-linear} with a scalar ${\alpha\in\mathbb{R}}$ for some $\alpha\neq 0$ iff $v_{1}=\alpha\cdot v_{2}$.
 Given a matrix $H$, and two trees $t_1$ and $t_2$, we say that $t_{1}\treesColin{\alpha}{H} t_{2}$ iff ${H}[t_{1}]$ and ${H}[t_{2}]$ are co-linear, with scalar $\alpha\neq 0$. That is, $H[t_{1}]=\alpha\cdot H[t_{2}]$. Note that if $H[t_{1}]=H[t_{2}]=\bar{0}$, then $t_{1}\treesColin{\alpha}{H} t_{2}$ for every $\alpha>0$.
We say that $t_{1}\treesEquiv{H} t_{2}$ if $t_{1}\treesColin{\alpha}{H} t_{2}$ for some $\alpha\neq 0$.  It is not hard to see that $\treesEquiv{H}$ is an equivalence relation.

The following proposition states that in the Hankel Matrix of an SUPCFG, the rows of trees that are rooted by the same non-terminal
 are co-linear.
\begin{prop}\label{prop:colinearity}
	Let $\hm{H}$ be the Hankel Matrix of an SUPCFG. 
	Let $t_{1},t_{2}$ be derivation trees rooted by the same non-terminal. Assume $\Prob(t_{1}),\Prob(t_{2})>0$.
	Then $t_{1}\treesColin{\alpha}{\hm{H}} t_{2}$ for some $\alpha\neq 0$.
\end{prop}

\begin{proof}
	Let $c$ be a context. Let $u\context v$ be the  yield of the context; that is, the letters with which the leaves of the context are tagged, in a left to right order (note that $u$ and $v$ might be $\varepsilon$). 
	We denote by $\Prob_{i}(c)$ the probability of deriving the given context $c$, while setting the context location to be $N_{i}$. That is:
	\begin{equation*}
	\Prob_{i}(c)=\Prob(S\underset{\mathcal{G}}{\overset{*}{\rightarrow}} u N_{i} v)
	\end{equation*}
	Let $\Prob(t_{1})$ and $\Prob(t_{2})$ be the probabilities for deriving the trees $t_{1}$ and $t_{2}$ respectively. Since the grammar is structurally unambiguous, we are guaranteed there is a single derivation tree for $uwv$ where $w$ is the yield of $t_1$ with the structure of $\contextConcat{c}{t_{1}}$ (and similarly for $t_2$), thus we obtain
	\begin{align*}
	\Prob(\contextConcat{c}{t_{1}})&=\Prob_{i}(c)\cdot \Prob(t_{1})\\
	\Prob(\contextConcat{c}{t_{2}})&=\Prob_{i}(c)\cdot \Prob(t_{2})
	\end{align*}
	Hence for every context $c$, assuming that $P_{i}(c)\neq 0$, we have
	\begin{equation*}
	\frac{\Prob(\contextConcat{c}{t_{1}})}{\Prob(\contextConcat{c}{t_{2}})}=\frac{\Prob(t_{1})}{\Prob(t_{2})}
	\end{equation*}
	For a context $c$ for which $\Prob_{i}(c)=0$ we obtain that $\Prob(\contextConcat{c}{t_{1}})=\Prob(\contextConcat{c}{t_{2}})=0$. Therefore, for every context
	\begin{equation*}
	\Prob(\contextConcat{c}{t_{1}})=\frac{\Prob(t_{1})}{\Prob(t_{2})}\Prob(\contextConcat{c}{t_{2}})
	\end{equation*}
	Thus $H[t_{1}]=\alpha\cdot H[t_{2}]$ for $\alpha=\frac{\Prob(t_{1})}{\Prob(t_{2})}$. Hence $H[t_{1}]$ and $H[t_{2}]$ are co-linear, and $t_{1}\treesColin{\alpha}{H} t_{2}$.
\end{proof}

We can thus conclude that the number of equivalence classes of $\equiv_{H}$ for an SUPCFG is finite and
bounded by the number of non-terminals plus one (for the zero vector). 
\begin{corollary}\label{cor:finite-equiv-cls}
	The skeletal tree-set for an SUPCFG has a finite number of equivalence classes under $\equiv_{H}$.
\end{corollary}

\begin{proof}
	Since the PCFG is structurally unambiguous, it follows that for every skeletal tree $s$ there is a single tagged parse tree $t$ such that $\skel(t)=s$. Hence, for every $s$ there is a single possible tagging, and a single possible non-terminal in the root. By Proposition \ref{prop:colinearity} every pair of trees $s_{1},s_{2}$  which are tagged by the same non-terminal, and in which $\Prob(s_{1}),\Prob(s_{2})>0$ are in the same equivalence class under $\equiv_{H}$. There is another equivalence class for all the trees $t\in\trees$ for which $\Prob(t)=0$. Since there is a finite number of non-terminals, there is a finite number of equivalence classes under $\equiv_{H}$.
\end{proof}

Next we would like to reveal the restrictions that can be imposed on a PMTA that corresponds to an SUPCFG.
We term a PMTA \emph{co-linear}, and denote it CMTA, if in every column of every transition matrix $\mu_\sigma$ there is at most one entry
which is non-negative. That is, in every transition matrix $\mu_\sigma$ of a CMTA 
all weights are either zero or positive, and there is at most one non-zero entry in each column.
\begin{prop}\label{prop:SUWCFGhaveCMTA}
	A CMTA can represent an SUPCFG.
\end{prop}

\commentout{
The proof relies on showing that a WCFG is structurally unambiguous iff it is invertible
and converting an invertible WCFG into a PMTA yields a CMTA.\footnote{A CFG $\grmr{G}=\langle \Vars,\Sigma,R,S\rangle$ is said to be
invertible if and only if $A \rightarrow \alpha$ and $B \rightarrow \alpha$ in $R$ implies $A = B$~\cite{Sakakibara92}.}
}

To prove Proposition~\ref{prop:SUWCFGhaveCMTA} we 
first show how to convert a PWCFG into a PMTA. Then we claim, that in case the PWCFG is structurally unambiguous
the resulting PMTA is a CMTA.

\paragraph{Converting a PWCFG into a PMTA}
Let $\tuple{\grmr{G},\theta}$ be a PWCFG where $\grmr{G}=\langle \Vars,\Sigma,R,S\rangle$. Suppose w.l.o.g that $\Vars=\{N_{0},N_{1},...,N_{|\Vars|-1}\}$, $\Sigma=\{\sigma_{0},\sigma_{1},...,\sigma_{|\Sigma|-1}\}$ and that $S=N_{0}$. Let $n=|\Vars|+|\Sigma|$. We define a function $\iota:\Vars\cup\Sigma\rightarrow\mathbb{N}_{\leq n}$ in the following manner:
\[
\iota(x) =
\begin{dcases*}
j
   & $x=N_{j}\in \Vars$\\
   |\Vars|+j
   & $x=\sigma_{j}\in \Sigma$
\end{dcases*}
\]
Note that since $\Vars\cap\Sigma=\emptyset$, $\iota$ is well defined. It is also easy to observe that $\iota$ is a bijection, so $\iota^{-1}:\mathbb{N}_{\leq n}\rightarrow \Vars\cup\Sigma$ is also a function.\\
We define a PMTA $\mathcal{A}_{\mathcal{G}}$ in the following manner:  
\begin{equation*}
    \mathcal{A}_{\mathcal{G}}=(\Sigma,\mathbb{R}_+,n,\mu,\lambda)
\end{equation*}
where $\lambda\in\mathbb{R}_*^n$ is defined as follows
\begin{equation*}
    \lambda=(1,0,...,0)
\end{equation*}
and for each $\sigma\in\Sigma$ we define
\[
\mu_\sigma[i] =
\begin{dcases*}
1
   & $i=\iota(\sigma)$\\
   0
   & otherwise
\end{dcases*}
\]

For $i\in\{1,\ldots,|\Vars|\}$ and $(i_1,i_2,\ldots,i_j)\in \{1,2\ldots,n\}^{|j|}$,
we define $R^{-1}(i,i_1,i_2,\ldots,i_j)$ to be the production rule
\[ \iota^{-1}(i)\longrightarrow \iota^{-1}(i_{1})~\iota^{-1}(i_{2})~\cdots~\iota^{-1}(i_{j}) \]
We define $\mu_?$ in the following way:
\[
{{\mu_?}^{i}}_{i_{1},...,i_{j}} =
\begin{dcases*}
\theta(R^{-1}(i,i_1,i_2,\ldots,i_j)) & $1\leq i\leq |\Vars|$\\
   0 & otherwise
\end{dcases*}
\]

We claim that the weights computed by the constructed PMTA
agree with the weights computed by the given grammar.

\begin{prop}\label{prop:equiv2}
For each skeletal tree ${t\in \skels(\deriv(\grmr{G}))}$ we have that $\mathcal{W}_{\mathcal{G}}(t)=\mathcal{A}_{\mathcal{G}}(t)$.
\end{prop}

\begin{proof}
The proof is reminiscent of the proof in the other direction, namely that of Proposition~\ref{prop:equiv1}.
We  first prove by induction that for each $t\in \skels(\deriv(\grmr{G})))$ the vector $\mu(t)=v=(v[1],v[2],...,v[n])$ calculated by $\mathcal{A}_{\mathcal{G}}$ maintains that for each $i\leq |\Vars|$, $v[i]=\grmrwgt{N_i}{t}$; and for $i>|\Vars|$ we have that $v[i]=1$ iff $t=\iota^{-1}(i)$ and $v[i]=0$ otherwise.

The proof is by induction on the height of $t$. For the base case $h=1$, thus $t$ is a leaf, therefore $t=\sigma\in\Sigma$. By definition $\mu_\sigma[i]=1$ if $i=\iota(\sigma)$ and $0$ otherwise. Hence $v[\iota(\sigma)]=1$, and for every $i\neq\iota(\sigma)$ $v[i]=0$. Since the root of the tree is in $\Sigma$, the root of the tree can't be a non-terminal, so $\grmrwgt{N_i}{t}=0$ for every $i$. Thus, the claim holds.

For the induction step,  $h>1$, thus $t=(? (t_{1}, t_{2},...,t_{k}))$ for some skeletal trees $t_{1}, t_{2},...,t_{k}$ of depth at most $h$. Let $\mu(t)=v=(v[1],v[2],...,v[n])$ be the vector calculated by $\mathcal{A}$ for $t$. By our definition of $\mu_?$, for every $i>|\Vars|$ ${{\mu_?}^{i}}_{i_{1},...,i_{j}}=0$ for all values of $i_{1},i_{2},...,i_{j}$. So for every $i>|\Vars|$ we have that $v[i]=0$ as required, since $t\notin\Sigma$. Now for $i\leq|\Vars|$, by definition of a multi-linear map we have that:
\begin{equation*}
        v[i]=\sum_{(i_1,i_2,\ldots,i_j)\in [|\Vars|]^j} {{\mu_?}^{i}}_{i_{1},...,i_{j}}\,v_{1}[j_{1}]\cdot...\cdot v_{j}[i_{j}]
\end{equation*}
Since $i\leq |\Vars|$, by our definition we have that:
\begin{equation*}
    {{\mu_?}^{i}}_{i_{1},...,i_{j}}\,=\theta(\iota^{-1}(i)\longrightarrow \iota^{-1}(i_{1})~\iota^{-1}(i_{2})~\cdots~\iota^{-1}(i_{j}))
\end{equation*}
For each $i_{k}$ let $B_{k}=\iota^{-1}(i_{k})$, also since $i\leq |\Vars|$, $\iota^{-1}(i)=N_{i}$, hence
\begin{equation*}
    {{\mu_?}^{i}}_{i_{1},...,i_{j}}\,=\theta(N_{i}\longrightarrow B_{1}B_{2}...B_{j})
\end{equation*}
For each $i$, by the induction hypothesis, if $t_{i}$ is a leaf, $v_{i}[j_{i}]=1$ only for $j_{i}=\iota(t_{i})$, and otherwise $v_{i}[j_{i}]=0$. If $t_{i}$ is not a leaf, then $v_{i}[j_{i}]=0$ for every $j_{i}>|\Vars|$; and for $j_{i}\leq |\Vars|$, we have that $v_{i}[j_{i}]=\grmrwgt{N_{j_i}}{t_i}$. Therefore we have:
\[
        v[i]=\sum_{(i_1,i_2,\ldots,i_j)\in [|\Vars|]^j} \begin{array}{l}
        \theta(N_{i}\rightarrow B_{1}B_{2}...B_{j})\cdot 
        \grmrwgt{N_{i_1}}{t_1}\cdots \grmrwgt{N_{i_j}}{t_j}
        \end{array}
\]
Thus by Lemma \ref{lem:weight} we have that $v[i]=\grmrwgt{N_{i}}{t}$ as required.

Finally, 
since $S=N_{0}$ and since by our claim, for each $i\leq |\Vars|$, $v_{i}=v[i]=\grmrwgt{N_{i}}{t}$, we get that $v[1]=\grmrwgt{S}{t}$. Also, since $\lambda=(1,0,...,0)$ we have that $\mathcal{A}_{\mathcal{G}}(t)$ is $v[1]$, which is $\grmrwgt{S}{t}$. Thus, it follows that  $\mathcal{W}_{\mathcal{G}}(t)=\mathcal{A}_{\mathcal{G}}(t)$
for every $t\in \skels(\deriv(\grmr{G}))$. 
\end{proof}

To show that the resulting PMTA is a CMTA we
need the following lemma, which makes use of the notion of \emph{invertible} grammars~\cite{Sakakibara92}.
A CFG $\grmr{G}=\langle \Vars,\Sigma,R,S\rangle$ is said to be
\emph{invertible} if and only if $A \rightarrow \alpha$ and $B \rightarrow \alpha$ in $R$ implies $A = B$.

\begin{lemma}\label{lem:invertible-iff-structuambig}
A CFG is invertible iff it is structurally unambiguous.
\end{lemma}
\begin{proof}
Let $\grmr{G}$ be a SUCFG. We  show that $\grmr{G}$ is invertible. Assume towards contradiction that there are derivations $N_{1}\rightarrow\alpha$ and $N_{2}\rightarrow\alpha$. Then the tree $?(\alpha)$ is structurally ambiguous since its root can be tagged by both $N_{1}$ and $N_{2}$.

For the other direction, let $\grmr{G}$ be an invertible grammar. We  show that $\grmr{G}$ is an SUCFG. Let $t$ be a skeletal tree. We show by induction on the height of $t$ that there is a single tagging for $t$.
For the base case, the height of $t$ is $1$. Therefore, $t$ is a leaf; so obviously, it has a single tagging.
For the induction step, we  assume that the claim holds for all skeletal trees of height at most $h\geq 1$. Let $t$ be a tree of height $h+1$. Then $t=?(t_{1},t_{2},...,t_{p})$ for some trees $t_{1},t_{2},...,t_{p}$ of smaller depth. By the induction hypothesis, for each of the trees $t_{1},t_{2},...,t_{p}$ there is a single possible tagging. 
Hence we have established that all nodes of $t$, apart from the root, have a single tagging.
Let $X_{i}\in\Sigma\cup N$ be the only possible tagging for the root of $t_{i}$. Let $\alpha=X_{1}X_{2}...X_{p}$. Since the grammar is invertible, there is a single non-terminal $N$ s.t. $N\rightarrow\alpha$. Hence, there is a single tagging for the root of $t$ as well. Thus $\grmr{G}$ is structurally unambiguous.
\end{proof}

We are finally ready to prove {Proposition~\ref{prop:SUWCFGhaveCMTA}}.

\begin{proof}[Proof of Prop.\ref{prop:SUWCFGhaveCMTA}]
    By Proposition~\ref{prop:equiv2} a WCFG $\tuple{\grmr{G},\theta}$ 
    can be represented by a PMTA $\aut{A}_{\grmr{G}}$, namely they provide the same weight for every
    skeletal tree. By Lemma~\ref{lem:invertible-iff-structuambig} the fact that $\grmr{G}$ is unambiguous implies
    it is invertible. We show that given $\grmr{G}$ is invertible, the 
    resulting PMTA is actually a CMTA. That is, in every column of the matrices of $\aut{A}_{\grmr{G}}$,
    there is at most one non-zero coefficient. Let $\alpha\in (\Sigma\cup \Vars)^{p}$, let $\iota(\alpha)$ be the extension of $\iota$ to $\alpha$, e.g., $\iota(aN_7bb)=\iota(a)\iota(N_7)\iota(b)\iota(b)$.
    Since $\grmr{G}$ is invertible, there is a single $N_{i}$ from which $\alpha$ can be derived, namely for which
     $\grmrwgt{N_{i}}{t^{N_i}_\alpha}>0$ where $t^{N_i}_\alpha$ is a tree deriving $\alpha$ with $N_i$ in the root.
     If $\alpha\in\Sigma$, i.e. it is a leaf, then we have that $\mu_\sigma[j]=0$ for every $j\neq i$, and $\mu_\sigma[i]>0$. If $\alpha\notin\Sigma$, then we have that ${{\mu_?}^{j}}_{\iota(\alpha)}=0$ for every $j\neq i$, and ${{\mu_?}^{i}}_{\iota(\alpha)}>0$, as required.
\end{proof}

\section{The Learning Algorithm}\label{sec:learning-CMTAs}
We are now ready to present the learning algorithm.
Let $\treesrs{T}:\trees(\Sigma)\rightarrow \R$ be an unknown tree series, and let $\hm{H}_\treesrs{T}$ be its Hankel Matrix.
The learning algorithm \LearnCMTA\ (or \cstar, for short), provided in Alg.~ \AlgLearnCMTA,
maintains a data structure called an \emph{observation table}.
An observation table for $\treesrs{T}$ is a quadruple $(T,C,H,B)$. Where $T\subseteq\trees(\Sigma)$ is a set of row titles, $C\subseteq\trees_{\context(\Sigma)}$ is a set of column titles,  $H:T\times C\rightarrow\mathbb{R}$ is a sub-matrix of $\hm{H}_\treesrs{T}$, and $B\subset T$, the so called \emph{basis}, is a set of row titles corresponding to rows of $H$ that are co-linearly independent. 

The algorithm starts with an almost empty observation table, where $T=\emptyset$, $C=\context$, $B=\emptyset$ and uses procedure $\Complete(T,C,H,B,\Sigma_0)$ to add the nullary symbols of the alphabet to the row titles, uses $\smq$ queries to fill in the table until certain criteria hold on the observation, namely it is \emph{closed} and \emph{consistent}, as defined in the sequel.
Once the table is closed and consistent, it is possible to extract from it a CMTA $\aut{A}$ (as we shortly explain). The algorithm then issues the query $\seq(\aut{A})$. If the result is ``yes'' the algorithm returns $\aut{A}$ which was determined to be structurally equivalent to the unknown series. Otherwise, the algorithm gets in return a counterexample $(s,\treesrs{T}(s))$, a structured string in the symmetric difference of $\aut{A}$ and $\treesrs{T}$, and its value. It then uses $\Complete$ to add all prefixes of $t$ to $T$ and uses \smq s to fill in the entries of the table until the table is once again closed and consistent.

\commentout{
	\begin{algorithm}
		\caption{$\LearnCMTA(T,C,H,B)$.}\label{alg:cstar}\label{alg:learn}
		\begin{algorithmic}[1]
			\State  Initialize $B\gets\emptyset,~T\gets\emptyset,~C\gets\{\context\}$ 
			\State $\Complete(T,C,H,B,\Sigma_0)$			
			\While{true}
			\State $\mathcal{A}\gets \ExtractCMTA(T,C,H,B)$
			\State $t\gets\seq(\mathcal{A})$
			\If{$t$ is null}
			\State \Return $\mathcal{A}$
			\EndIf
			\State $\Complete(T,C,H,B,  \pref(t))$
			\EndWhile
		\end{algorithmic}
	\end{algorithm}
	}

\begin{figure}[t]
\noindent\makebox[.98\textwidth]{
	\includegraphics[scale=0.99,page=2, trim=1cm 22cm 10cm 1.8cm]{figures.pdf}
}
\\
\noindent\makebox[.98\textwidth]{
	\includegraphics[scale=0.99,page=3, clip, trim=1cm 23.6cm 10cm 1.8cm]{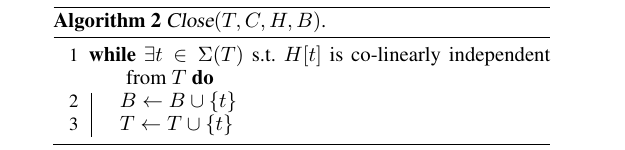}
}
\\
\noindent\makebox[.98\textwidth]{
	\includegraphics[scale=0.99,page=4, clip, trim=1cm 22cm 10cm 1.8cm]{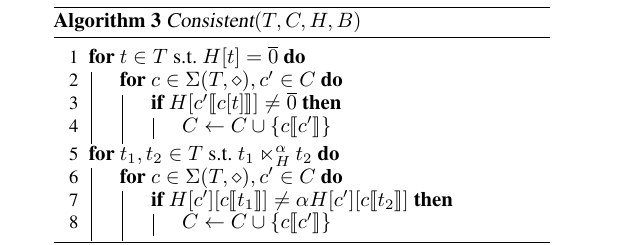}
}
\\
\noindent\makebox[.98\textwidth]{
	\includegraphics[scale=0.99,page=5, clip, trim=1cm 23.6cm 10cm 1.8cm]{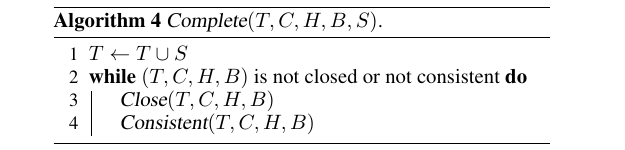}
}
\end{figure}

Given a set of trees $T$ we use $\Sigma(T)$ for the set of trees $\{ \sigma(t_1,\ldots,t_k) ~|~\exists {\Sigma_k\in\Sigma},\ {\sigma\in\Sigma_k},$ $\ {t_i\in T},\ {\forall 1\leq i \leq k}\}$.
The procedure $\Close(T,C,H,B)$, given in Alg.~\AlgCloseCMTA, checks if $\hm{H}[t][C]$
is co-linearly independent from $T$ for some tree $t\in\Sigma(T)$. If so it adds $t$ to both $T$ and $B$ and loops back until no such trees are found, in which case  the table is termed \emph{closed}.

We use $\Sigma(T,t)$ for the set of trees in $\Sigma(T)$ satisfying that one of the children is the tree $t$. 
We use $\Sigma(T,\context)$ for the set of contexts all of whose children (but the context) are in $T$.
An observation table $(T,C,H,B)$ is said to be \emph{zero-consistent} if for every
tree $t\in T$ for which $H[t]=\overline{0}$ it holds that $H[\contextConcat{c}{t'}]=\overline{0}$ for every $t'\in \Sigma(T,t)$ and $c\in C$ (where $\overline{0}$ is a vector of all zeros). It is said to be \emph{co-linear consistent} if for every $t_1,t_2\in T$ such that $t_1\treesColin{\alpha}{H} t_2$ 
and every context $c\in \Sigma(T,\context)$ we have that   $\contextConcat{c}{t_1}\treesColin{\alpha}{H} \contextConcat{c}{t_2}$. The procedure $\Consistent$, given in Alg.~\AlgConsistentCMTA, looks for trees which violate the zero-consistency or co-linear consistency requirement, and for every violation, the respective context is added to the set of columns $C$. 

The procedure $\Complete(T,C,H,B,S)$, given in Alg.~\AlgCompleteCMTA, first adds the trees in $S$ to $T$, 
then runs procedures $\Close$ and $\Consistent$ iteratively until the table is both closed and consistent. When the table is closed and consistent the algorithm extracts from it a CMTA as detailed in Alg.~\AlgExtractCMTA.

The procedure $\ExtractCMTA$, given in Alg.~\AlgExtractCMTA, sets the output vector $\lambda$ of the CMTA by
setting its $j$-th coordinate to $H(b_j,\diamond)$ (lines 17-18).
For each letter $\sigma\in\Sigma_k$ it builds the $d\times d^k$ matrix $\mu_\sigma$ (where $d$ is the size of the basis $B$)
by setting its coefficients as follows (lines 4-15). For each possible assignments of trees from the basis as children of $\sigma$,
namely for each $(i_1,i_2,\ldots,i_k)\in\{1,2,\ldots,d\}$ it inspects the row $H(t)$ for $t=\sigma(b_{i_1},b_{i_2},\ldots,b_{i_k})$.
If this row is zero then the column of $\mu_\sigma$ corresponding to $(i_1,i_2,\ldots,i_k)$ is set to zero.
Otherwise, since the basis consists of rows that are co-linearly independent, there exists a single row in the basis, $b_i$, such that the row $H(t)$ for $t=\sigma(b_{i_1},b_{i_2},\ldots,b_{i_k})$ is co-linear to $b_i$. Let $\alpha\in\mathbb{R}_+$ be such that  $t\treesColin{\alpha}{H} b_i$. Then the entry $\mu^i_{i_1,i_2,\ldots,i_k}$ is set to $\alpha$ and the entries $\mu^j_{i_1,i_2,\ldots,i_k}$ for $j\neq i$ are set to zero.

\begin{figure}[t]
\noindent\makebox[.98\textwidth]{
	\includegraphics[scale=0.91,page=6, clip, trim=1cm 16.9cm 10cm 1.8cm]{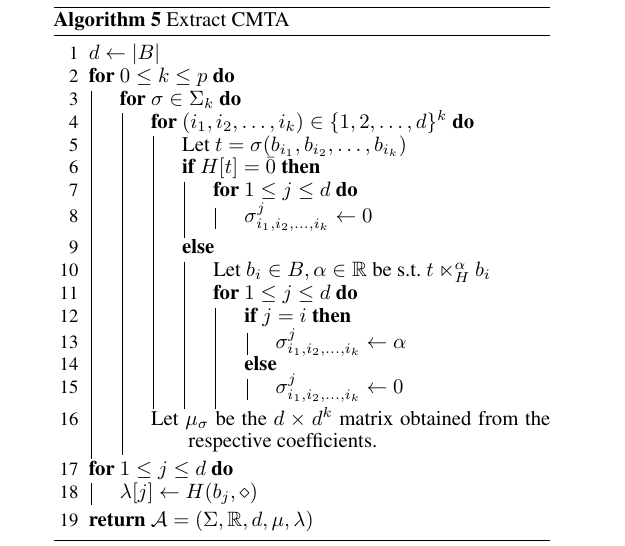}
}
\end{figure}

Overall we can show that the algorithm $\LearnCMTA$ always terminates, returning a correct CMTA whose dimension is minimal, namely it equals the rank $n$ of the Hankel matrix for the target language.
It does so while asking at most $n$ equivalence queries, and the number of membership queries is polynomial in $n$, and in the size of the largest counterexample $m$, but of course exponential in $p$, the highest rank of a symbol in $\Sigma$. Hence for a grammar in Chomsky Normal Form, where $p=2$, it is polynomial in all parameters.

\begin{theorem}\label{thm:bounds}
	Let $n$ be the rank of the target language, let $m$ be the size of the largest counterexample given by the teacher, and let $p$ be the highest rank of a symbol in $\Sigma$. Then the algorithm makes at most $n\cdot(n+m\cdot n+|\Sigma|\cdot (n+m\cdot n)^{p})$ $\smq$s and at most $n$ $\seq$s.
\end{theorem}

We first give a running example in \S\ref{sec:running-example}, and then prove this theorem in \S\ref{sec:correctness}.

\subsection{Running Example}\label{sec:running-example}

We will now demonstrate a running example of the learning algorithm. 
For the unknown target  consider the series which gives probability $(\frac{1}{2})^n$ to
strings of the form $a^nb^n$ for $n\geq 1$ and probability zero to all other strings. 
 This series can be generated by the following SUPCFG $\grmr{G}=\langle \Vars,\{a,b\},R,S\rangle$
 with $\Vars=\{S,S_{2}\}$, and the following derivation rules:
$$\begin{array}{l@{\ \ \longrightarrow\ \ }lll}
S& aS_{2}~[\frac{1}{2}]\ \mid\ ab~[\frac{1}{2}]\\
S_{2}& Sb~~~[1]
\end{array}$$

The algorithm initializes $T=\{a,b\}$ and $C=\{\context\}$, fills in the entries of $M$ using $\smq$s, first for the rows of $T$ and then for their one letter extensions $\Sigma(T)$ (marked in gray), resulting in the observation table in Tab.~\ref{table:obs1}~(a).

We can see that the table is not closed, since $?(a,b)\in\Sigma(T)$ but is not co-linearly spanned by $T$, so we add it to $T$. Also, the table is not consistent, since $a\treesColin{1}{H} b$, but $\mq(\contextConcat{\context}{?(a,b)})\neq\mq(\contextConcat{\context}{?(a,a)})$, so we add $?(a,\context)$ to $C$, and we obtain the observation table in Tab.~\ref{table:obs2}~(b).
From now on we  omit $0$ rows of $\Sigma(T)$ for brevity.

\begin{table}
{\small{
\begin{tabular}{c@{\qquad}c@{\qquad}c}
	\begin{tabular}{ l | c}
		\multicolumn{2}{c}{\vspace*{-9.4pt}}\\
		& \rotatebox[origin=c]{90}{$~~\Tree [.$~\context~~~$ ]$} \\ \hline
		$a$& $0$\\ \hline
		$b$& $0$\\ \hline
		\color{gray} $?(a,a)$ & $0$\\\hline
		\color{gray} $?(a,b)$ & $0.5$\\\hline
		\color{gray} $?(b,a)$ & $0$\\\hline
		\color{gray} $?(b,b)$ & $0$\\\hline
		\multicolumn{2}{c}{} \\
	\end{tabular}
&
	\begin{tabular}{l | l | c | c }
	\multicolumn{4}{c}{}\\
		& & \rotatebox[origin=c]{90}{$~\Tree [.$\context$ ]$} & \rotatebox[origin=c]{90}{$~?(a,\context)$} \\ \hline
		$t_1$ & $a$ & $0$ & $0$ \\ \hline
		$t_2$ & $b$ & $0$ & $0.5$\\ \hline
		$t_3$ & $?(a,b)$ & $0.5$ &$ 0$ \\\hline
		& \color{gray} $?(?(a,b),b)$ & $0$ & $0.25$ \\ \hline
		\multicolumn{4}{c}{}\\
	\end{tabular}
&
	\begin{tabular}{l | l | l | l | l }
		& & \Tree [.$\context$ ] & \rotatebox[origin=c]{90}{$~?(a,\context)$} & \rotatebox[origin=c]{90}{$~?(\context,b)$}\\ \hline
		$t_1$ & $a$ & $0$ & $0$ & $0.5$ \\ \hline
		$t_2$ & $b$ & $0$ & $0.5$ & $0$ \\ \hline
		$t_3$ & $?(a,b)$ & $0.5$ & $0$ & $0$ \\\hline
		$t_4\notin B$ & $?(a,a)$ & $0$ & $0$ & $0$ \\\hline
		& \color{gray} $?(?(a,b),b)$ & $0$ & $0.25$ & $0$ \\ \hline
		\multicolumn{5}{c}{}\\
	\end{tabular}
\\
\textbf{(a)} & \textbf{(b)} & \textbf{(c)}
\end{tabular}
}}
\caption{Observation tables (a) (b) and (c)}\label{table:obs1}\label{table:obs2}\label{table:obs3}
\end{table}

The table is now closed but it is not zero-consistent, since we have $H[a]=\overline{0}$, but there exists a context with children in $T$, specifically $?(\context,b)$,
with which when $a$ is extended the result is not zero, namely $H[?(a,b)]\neq 0$. So we add this context
and we obtain the observation table in Tab.~\ref{table:obs3}~(c).

Note that $t_{4}$ was added to $T$ since it was not  spanned by $T$, but it is not a member of $B$, since $H[t_{4}]=0$. 
We can extract the following CMTA $\aut{A}_1=(\Sigma,\mathbb{R},d,\mu,\lambda)$ of dimension $d=3$ since $|B|=|\{t_{1},t_{2},t_{3}\}|=3$.
Let $\V=\R^3$.  For the letters $\sigma\in \Sigma_0=\{a,b\}$ we have that $\mu_\sigma:\V^0\rightarrow \V$, namely $\mu_a$ and $\mu_b$ are $3\times 3^0$-matrices. Specifically, following Alg.~\AlgCompleteCMTA we get that $\mu_{a}=(1,0,0)$, $\mu_{b}=(0,1,0)$ as $a$ is the first element of $B$ and $b$ is the second.
For ${?}\in\Sigma^2$ we have that $\mu_{?}:\V^2\rightarrow \V$, thus $\mu_{?}$ is a $3\times 3^2$-matrix.
We compute the entries of $\mu_{?}$ following Alg.~\AlgCompleteCMTA. For this, we consider all pairs of indices $(j,k)\in\{1,2,3\}^2$. For each such entry we look for the row $t_{j,k}=?(t_j,t_k)$
and search for the base row $t_i$ and the scalar $\alpha$ for which  $t_{j,k} \treesColin{\alpha}{H} t_i$.
We get that $t_{1,2}\treesColin{1}{H} t_3$, $t_{3,2}\treesColin{0.5}{H} t_2$ and for all other $j,k$ we get $t_{j,k}\treesColin{1}{H} t_{4}$, so we set $c^{i}_{j,k}$ to be $0$ for every $i$. Thus, we obtain the following matrix for $\mu_{?}$
$$
\eta_{?}=\begin{bmatrix}0&0&0&0&0&0&0&0&0\\0&0&0&0&0&0&0&0.5&0\\
0&1&0&0&0&0&0&0&0\end{bmatrix}$$
 The vector $\lambda$ is also computed via Alg.~\AlgCompleteCMTA, and we get $\lambda=(0,0,0.5)$. 
 
 The algorithm now asks an equivalence query and receives the tree $p$ given in Fig.~\ref{fig:tree-p} (i) as a counterexample:
\begin{figure} 
\begin{tabular}{c@{\qquad\qquad}c}
 \Tree [.? $a$ [.? [.? $a$ $b$ ] [.? [.? $a$ $b$ ] $b$ ] ] ] 
 &
 \scalebox{.8}{
   \Tree [.$(0,0,0.25)$ $(1,0,0)$ [.$(0,0.25,0)$ [.$(0,0,1)$ $(1,0,0)$ $(0,1,0)$ ] [.$(0,0.5,0)$ [.$(0,0,1)$ $(1,0,0)$ $(0,1,0)$ ] $(0,1,0)$ ] ] ]
}
 \\
 \textbf{(i)} & \textbf{(ii)}
 \end{tabular}
 \caption{(i) The tree $p$  and (ii) the tree $p$ annotated with the value $\mu(t)$ for its of sub-trees $t$. }\label{fig:tree-p}\label{fig:mu-tree-p}
 \end{figure}
 
 Indeed, while $\smq(p)=0$ we have that $\aut{A}(p)=0.125$. To see why $\aut{A}(p)=0.125$, let's look at the values $\mu(t)$ for every sub-tree $t$ of $p$. For the leaves, we have $\mu(a)=(1,0,0)$ and $\mu(b)=(0,1,0)$.
 
 Now, to calculate $\mu(?(a,b))$, we need to calculate $\mu_{?}(\mu(a),\mu(b))$. To do that, we first compose them as explained in the \emph{multilinear functions} paragraph of Sec.~\ref{sec:mta}, see also Fig.~\ref{fig:MTA}. The vector $P_{\mu(a),\mu(b)}$ is $(0,1,0,0,0,0,0,0,0)$. When multiplying this vector by the matrix $\eta_{?}$ we obtain $(0,0,1)$. So $\mu(?(a,b))=(0,0,1)$. Similarly, to obtain  $\mu(?(?(a,b),a))$ we first compose the value $(0,0,1)$
 for $?(a,b)$ with the value $(0,1,0)$ for $a$ and obtain $P_{\mu?(a,b),\mu(a)}=(0,0,0,0,0,0,0,1,0)$. 
 Then we multiply $\eta_{?}$ by $P_{\mu?(a,b),\mu(a)}$ and obtain $(0,0.5,0)$. In other words, 
 $$\mu(?(?(a,b),a))=\mu_{?}\left(\left[\begin{smallmatrix}0\\ 0\\ 1\end{smallmatrix}\right],\left[\begin{smallmatrix}0\\ 1\\ 0\end{smallmatrix}\right]\right)=\left[\begin{smallmatrix}0\\0.5\\0\end{smallmatrix}\right].$$ 
 The tree in Fig.~\ref{fig:mu-tree-p} (ii) depicts the entire calculation by marking the values obtained for each sub-tree. We can see that  $\mu(p)=(0,0,0.25)$, thus we get that $\aut{A}=\mu(p)\cdot\lambda=0.125$.
 
 We add all prefixes of this counterexample to $T$ and we obtain the observation table in Tab.~\ref{table:obs4}~(d).
 \begin{table}
 \hspace*{-15pt}
 {\small{
 \begin{tabular}{c@{\qquad}c}
 	\begin{tabular}{l | l | c | c | c }
 		 & & \Tree [.$\context$ ] & \rotatebox[origin=c]{90}{$?(a,\context)$} & \rotatebox[origin=c]{90}{$?(\context, b)$} \\ \hline
 		$t_1$ & $a$  & $0$ & $0$ & $\frac{1}{2}$ \\ \hline
 		$t_2$ & $b$ & $0$ & $\frac{1}{2}$ &$0$ \\ \hline
 		$t_3$ & $?(a,b)$ & $\frac{1}{2}$ & $0$ & $0$ \\\hline
 		$t_4$ & $?(a,a)$ & $0$ & $0$ & $0$\\\hline
 		$t_5$ & $?(?(a,a),a)$ & $0$ & $0$ & $0$\\\hline
 		$t_6$ & $?(?(a,b),b)$ & $0$ & $\frac{1}{4}$ & $0$ \\\hline
 		$t_7$ & $?(?(a,b),?(?(a,b),b))$ & $0$ & $0$ & $0$\\\hline
 		$t_8$ & $?(a,?(?(a,b),?(?(a,b),b)))$ & $0$ & $0$ & $0$ \\ \hline
 		\multicolumn{5}{c}{}\\
 	\end{tabular}
 &
 	\begin{tabular}{l | l | c | c | c | c}
 		& & \Tree [.$\context$ ] & \rotatebox[origin=c]{90}{$?(a,\context)$} & \rotatebox[origin=c]{90}{$?(\context, b)$} & \rotatebox[origin=c]{90}{$?(a,?(?(a,b),\context))$}\\ \hline
 		$t_1$ & $a$ & $0$ & $0$ & $\frac{1}{2}$ & $0$\\ \hline
 		$t_2$ & $b$ & $0$ & $\frac{1}{2}$ & $0$ & $\frac{1}{4}$ \\ \hline
 		$t_3$ & $?(a,b)$ & $\frac{1}{2}$& $0$ & $0$ & $0$ \\\hline
 		$t_4$ & $?(a,a)$ & $0$ & $0$ & $0$ & $0$ \\\hline
 		  & $?(?(a,a),a)$ & $0$ & $0$ & $0$ & $0$ \\\hline
 		$t_5$ & $?(?(a,b),b)$ & $0$ & $\frac{1}{4}$ & $0$ & $0$ \\\hline
 		  & $?(?(a,b),?(?(a,b),b))$ & $0$ & $0$ & $0$ & $0$\\\hline
 		  & $?(a,?(?(a,b),?(?(a,b),b)))$ & $0$ & $0$ & $0$ & $0$ \\\hline
 		 $t_6$  & \color{gray}$?(a,?(?(a,b),b))$ & $\frac{1}{4}$  & $0$ & $0$ & $0$ \\ \hline
 		 \multicolumn{5}{c}{}\\
 	\end{tabular}\\
 	\textbf{(d)} & \textbf{(e)}
\end{tabular}
}}
\caption{Observation tables (d) and (e)}\label{table:obs4}\label{table:obs5}
\end{table}
 This table is not consistent since while $t_6 \treesColin{0.5}{H} t_2$ this co-linearity is not preserved when extended with $t_{3}=?(a,b)$ to the left, as evident from the context $?(a,\context)$.
We thus add the context $\contextConcat{?(a,\context)}{?(?(a,b),\context)}=?(a,?(?(a,b),\context))$ to obtain the final observation table given in in Tab.~\ref{table:obs5}~(e).

 The table is now closed and consistent, and we extract the following CMTA from it: 
 $\aut{A}_3=(\Sigma,\mathbb{R},4,\mu,\lambda)$ with $\mu_{a}=(1,0,0,0)$, $\mu_{b}=(0,1,0,0)$. 
 Now $\mu_{?}$ is a $4\times 4^2$ matrix. Its non-zero entries are $c^3_{1,2}=1$, $c^4_{3,2}=1$ and $c^3_{1,4}=\frac{1}{2}$
 since $t_{1,2}\treesColin{1}{H} t_3$, $t_{3,2} \treesColin{1}{H} t_5$, $t_{1,5} \treesColin{1}{H} t_6 \treesColin{0.5}{H} t_3$.
 And  for every other combination of unit-basis vectors we have $t_{i,j}\treesColin{1}{H} t_4$.   
 The final output vector is $\lambda=(0,0,0.5,0)$.
 
 The equivalence query on this CMTA returns true, hence the algorithm now terminates, and we can convert this CMTA into a WCFG. 
 Applying the transformation provided in Fig.~\ref{fig:eqs-pmta-to-pcfg} we obtain the following WCFG:
 \begin{align*}
 S&\longrightarrow N_{3}~[0.5]   \\
 N_{1}&\longrightarrow a~[1.0 ]  \\
 N_{2}&\longrightarrow b~[1.0]   \\
 N_{3}&\longrightarrow N_{1} N_{2}~[1.0]\ \  |\ \ N_{1} N_{4}~[0.5]   \\
 N_{4}&\longrightarrow N_{3} N_{2}~[1.0]  
 \end{align*}
 
 Now, following~\cite{abney1999relating,smith2007weighted}
 we can calculate the partition functions for each non-terminal. Let $f_{N}$ be the sum of the weights of all trees whose root is $N$, we obtain:
 $$\begin{array}{l@{\quad}l@{\quad}l@{\quad}l@{\quad}l}
 f_{S}=1, & 
 f_{N_{1}}=1, &
 f_{N_{2}}=1, &
 f_{N_{3}}=2, &
 f_{N_{4}}=2. 
 \end{array}$$
 
 Hence we obtain the PCFG 
 \begin{align*}
 S&\longrightarrow N_{3}~[1.0]\\
 N_{1}&\longrightarrow a~[1.0]\\
 N_{2}&\longrightarrow b~[1.0]\\
 N_{3}&\longrightarrow N_{1} N_{2}~[0.5]\ |\ N_{1} N_{4}~[0.5]\\
 N_{4}&\longrightarrow N_{3} N_{2}~[1.0]
 \end{align*}
 which is a correct grammar for the unknown probabilistic series.

The careful reader might notice that the produced grammar is not the most succinct one for this language.  
Indeed, the conversation of a CMTA to a PCFG could be optimized; currently it adds additional $O(\Sigma)$ non-terminals and derivation rules (for every terminal $a$ a non-terminal $N_a$ and a respective production rule $N_a \longrightarrow a$ are introduced).
    Algorithm 1 itself does return the minimal CMTA with respect to the given skeletal trees.

\subsection{Correctness Proof}\label{sec:correctness}
To prove the main theorem we require a series of lemmas, which we state and prove here.
We start with some additional notations.
Let $v$ be a row vector in a given matrix. Let $C$ be a set of columns. We denote by $v[C]$ the restriction of $v$ to the columns of $C$. For a set of row-vectors $V$ in the given matrix, we denote by $V[C]$ the restriction of all vectors in $V$ to the columns of $C$.

\begin{lemma}\label{lemmaContextColin}
	Let $B$ be a set of row vectors in a matrix $H$, and let $C$ be a set of columns. If a row $v[C]$ is co-linearly independent from $B[C]$ then $v$ is co-linearly independent from $B$.
\end{lemma}
\begin{proof}
	Assume towards contradiction that there is a vector $b\in B$ and a scalar $\alpha\in\reals$ s.t. $v=\alpha b$. Then for every column $c$ we have $v[c]=\alpha b[c]$. In particular that holds for every $c\in C$. Thus, $v[C]=\alpha b[C]$ and so $v[C]$ is not co-linearly independent from $B[C]$, contradicting our assumption.
\end{proof}

\begin{lemma}\label{replacementLemma}
	Let $\aut{A}=(\Sigma,\mathbb{R},d,\mu,\lambda)$ 
	be a CMTA. Let $t_{1},t_{2}$ s.t. $\mu(t_{1})=\alpha\cdot\mu(t_{2})$. Then for every context $c$
	\begin{equation*}
	\mu(\contextConcat{c}{t_{1}})=\alpha\cdot\mu(\contextConcat{c}{t_{2}})
	\end{equation*}
\end{lemma}	
	\begin{proof}
	The proof is by induction on the depth of $\context$ in $c$.
	For the base case, 	the depth of $\context$ in $c$ is $1$. Hence, $c=\context$ and indeed we have
			\begin{equation*}
			\mu(\contextConcat{c}{t_{1}})=\mu(t_{1})=\alpha\cdot\mu(t_{2})=\alpha\cdot\mu(\contextConcat{c}{t_{2}})
			\end{equation*}
			as required.

		For the induction step, 
		 assume the claim holds for all contexts where $\context$ is in depth at most $h$. Let $c$ be a context s.t. $\context$ is in depth $h+1$. Hence, there exists contexts $c_1$ and $c_2$ s.t.  $c=\contextConcat{c_{1}}{c_{2}}$ where $c_{2}=\sigma(s_{1},s_{2},...,s_{i-1},\context,s_{i+1},...,s_{p})$ for some $s_i$'s and the depth of $\context$ in $c_{1}$ is $h$. Let $t'_{1}=\contextConcat{c_{2}}{t_{1}}$ and let $t'_{2}=\contextConcat{c_{2}}{t_{2}}$. We have
		\begin{align*}
		\mu(t'_{1})&=
		\mu(\contextConcat{c_{2}}{t_{1}})\\
		&=\mu_{\sigma}(\mu(s_{1}),\mu(s_{2}),...,\mu(s_{i-1}),\mu(t_{1}),\mu(s_{i+1}),...,\mu(s_{p}))\\
		&=\mu_{\sigma}(\mu(s_{1}),\mu(s_{2}),...,\mu(s_{i-1}),\alpha\cdot\mu(t_{2}),\mu(s_{i+1}),...,\mu(s_{p}))
		\end{align*}
		Similarly for $t_{2}$ we obtain
		\begin{equation*}
		\mu(t'_{2})=\mu_{\sigma}(\mu(s_{1}),\mu(s_{2}),...,\mu(s_{i-1}),\mu(t_{2}),\mu(s_{i+1}),...,\mu(s_{p})).\\
		\end{equation*}
		By properties of multi-linear functions we obtain:
		\begin{equation*}
		\begin{array}{rl}
		\mu(\sigma(s_{1},s_{2},...,s_{i-1},t_{1},s_{i+1},...,s_{p}))&=\\
		\alpha\cdot\mu(\sigma(s_{1},s_{2},...,s_{i-1},t_{2},s_{i+1},...,s_{p}))
		\end{array}
		\end{equation*}

		Thus, $\mu(t'_{1})=\alpha\cdot\mu(t'_{2})$, and by the induction hypothesis on $c_{1}$ we have:
		\begin{equation*}
		\mu(\contextConcat{c_{1}}{t'_{1}})=\alpha\cdot\mu(\contextConcat{c_{1}}{t'_{2}})
		\end{equation*}
		Hence
		\begin{equation*}
		\mu(\contextConcat{c}{t_{1}})=\mu(\contextConcat{c_{1}}{t'_{1}})=\alpha\cdot\mu(\contextConcat{c_{1}}{t'_{2}})=\alpha\cdot\mu(\contextConcat{c}{t_{2}})
		\end{equation*}
		as required.	
\end{proof}

Recall that a subset $B$ of $T$ is called a \emph{basis} if for every ${t\in T}$, if ${H[t]\neq 0}$ then there is a unique $b\in B$ such that $t\treesColin{\alpha}{H} b$.
Let $(T,C,H,B)$ be an observation table. Then $B=\{b_1,b_2\ldots,b_d\}$ is a {basis} for $T$, and if $b_i$ is the unique element of $B$ for which $t\treesColin{\alpha}{H} b_i$ for some $\alpha$, 
we say that $\classrepr{t}{B}=b_{i}$, $\classcoeff{t}{B}=\alpha$, and $\classind{t}{B}=i$. 
The following lemma states that the value assigned to a tree $?(t_{1},t_{2},...,t_{p})$ all of whose children are in $T$, 
can be computed by multiplying the respective coefficients $\classcoeff{t_i}{B}$ witnessing the co-linearity of $t_i$ to its respective
base vector $\classrepr{t_i}{B}$.

\begin{lemma}\label{consistencyExpansion}
	Let $(T,C,H,B)$  be a closed consistent observation table. Let $t_{1},t_{2},...,t_{p}\in T$, and let $t=?(t_{1},t_{2},...,t_{p})$. Then
	\begin{equation*}
	H[?(t_{1},t_{2},...,t_{p})]=\prod_{i=1}^{p}\classcoeff{t_{i}}{B}\cdot H[?(\classrepr{t_{1}}{B},\classrepr{t_{2}}{B},...,\classrepr{t_{p}}{B})]
	\end{equation*}
\end{lemma}

	\begin{proof}
		Let $k$ be the number of elements in $t_{1},t_{2},...,t_{p}$  such that $t_{i}\neq\classrepr{t_{i}}{B}$. We proceed by induction on $k$. 
		For the base case, we have $k=0$, so for every $t_{i}$ we have $t_{i}=\classrepr{t_{i}}{B}$ and $\classcoeff{t_i}{B}=1$. Hence, obviously we have
		\begin{equation*}
		H[?(t_{1},t_{2},...,t_{p})]=\prod_{i=1}^{p}\classcoeff{t_{i}}{B}\cdot H[?(\classrepr{t_{1}}{B},\classrepr{t_{2}}{B},...,\classrepr{t_{p}}{B})]
		\end{equation*}
		
		Assume now the claim holds for some ${k\geq 0}$. Since ${k+1>0}$ there is at least one $i$  such that  $t_{i}\neq\classrepr{t_{i}}{B}$. Let $t'=?(t_{1},t_{2},...,t_{i-1},\classrepr{t_{i}}{B},t_{i+1},...,t_{p})$. Since the table is consistent, we have that $H[t]=\classcoeff{t_{i}}{B}\cdot H[t']$.
		Now,  $t'$ has $k$ children  such that  $t_{i}\neq\classrepr{t_{i}}{B}$, so from the induction hypothesis we have
		\begin{equation*}
		H[t']=\prod_{\begin{array}{c}{j=1}\\{j\neq i}\end{array}}^{p}\classcoeff{t_{j}}{B}\cdot H[?(\classrepr{t_{1}}{B},\classrepr{t_{2}}{B},...,\classrepr{t_{p}}{B})
		\end{equation*}
		Hence we have
		\begin{equation*}
		H[t]=\classcoeff{t_{i}}{B}\cdot H[t']=\prod_{j=1}^{p}\classcoeff{t_{j}}{B}\cdot H[?(\classrepr{t_{1}}{B},\classrepr{t_{2}}{B},...,\classrepr{t_{p}}{B})
		\end{equation*}
		as required.
	\end{proof}

The following lemma states that if $t=?(t_{1},t_{2},...,t_{p})$  is co-linear to $s=?(s_{1},s_{2},...,s_{p})$ 
and $t_i$ is co-linear to $s_i$, for every $1\leq i\leq p$ and $H[t]\neq 0$ then the ratio between the tree coefficient and the product of its children coefficients is the same.
\begin{lemma}\label{lem:equal-prods-coef}
		Let $t=?(t_{1},t_{2},...,t_{p})$ and $s=?(s_{1},s_{2},...,s_{p})$ s.t. $t_{i}\treesEquiv{H} s_{i}$ for $1\leq i\leq p$. 
		Then $$\frac{\classcoeff{t}{B}}{\prod_{i=1}^{p}\classcoeff{t_{i}}{B}}=\frac{\classcoeff{s}{B}}{\prod_{i=1}^{p}\classcoeff{s_{i}}{B}}$$
\end{lemma}

\begin{proof}
 Let $t'=?(\classrepr{t_{1}}{B},\classrepr{t_{2}}{B} ...,\classrepr{t_{p}}{B})$.
 Note that we also have $t'=?(\classrepr{s_{1}}{B},\classrepr{s_{2}}{B},...,\classrepr{s_{p}}{B})$.
  Then from Lemma \ref{consistencyExpansion} we have that $H[t]=\prod_{i=1}^{p}\classcoeff{t_{i}}{B}\cdot H[t']$. 
  Similarly we have that $H[s]=\prod_{i=1}^{p}\classcoeff{s_{i}}{B}\cdot H[t']$. 
  It follows from Lemma~\ref{consistencyExpansion} that $t\treesEquiv{H} s$, since for each $i$ $t_i\treesEquiv{H} s_i$. 
  Let $b=\classrepr{t}{B}=\classrepr{s}{B}$. We have $H[t]=\classcoeff{t}{B}\cdot H[b]$, and $H[s]=\classcoeff{s}{B}\cdot H[b]$.
Thus we have
\begin{equation*}
\classcoeff{t}{B}\cdot H[b]=\prod_{i=1}^{p}\classcoeff{t_{i}}{B}\cdot H[t']
\end{equation*}
and
\begin{equation*}
\classcoeff{s}{B}\cdot H[b]=\prod_{i=1}^{p}\classcoeff{s_{i}}{B}\cdot H[t']
\end{equation*}
Hence we have
\begin{equation*}
\frac{\classcoeff{t}{B}}{\prod_{i=1}^{p}\classcoeff{t_{i}}{B}}\cdot H[b]=H[t']=\frac{\classcoeff{s}{B}}{\prod_{i=1}^{p}\classcoeff{s_{i}}{B}}\cdot H[b]
\end{equation*}
Since $H[t]\neq 0$, and $t\treesEquiv{H} b\treesEquiv{H} t'$ we obtain that $H[b]\neq 0$ and $H[t']\neq 0$. 
Therefore
\begin{equation*}
\frac{\classcoeff{t}{B}}{\prod_{i=1}^{p}\classcoeff{t_{i}}{B}}=\frac{\classcoeff{s}{B}}{\prod_{i=1}^{p}\classcoeff{s_{i}}{B}} \qedhere
\end{equation*}
\end{proof}

The next lemma relates the value $\mu(t)$ to  $t$'s coefficeint, $\classcoeff{t}{B}$, and the vector for respective row in the basis, $\classind{t}{B}$.

\begin{lemma}\label{lem:mu-t-rel}
	Let $t\in T$. If $H[t]\neq 0$ then $\mu(t)=\classcoeff{t}{B}\cdot [t]$. If $H[t]=0$ then $\mu(t)=0$.
\end{lemma}
\begin{proof}
	The proof is by induction on the height of $t$.
	For the base case, $t=\sigma$ is a leaf, for some $\sigma\in\Sigma$. If $H[t]\neq 0$, by Alg.~\AlgExtractCMTA, 
	we set $\sigma^{\classind{t}{B}}$ to be $\classcoeff{t}{B}$, and for every $j\neq\classind{t}{B}$ we set $\sigma^{j}$ to be $0$, 
	so $\mu(t)=\mu_{\sigma} =\classcoeff{t}{B}\cdot [t]$ as required. 
	Otherwise, if $H[t]=0$ then we set $\sigma^{i}$ to be $0$ for every $i$, so $\mu(t)=0$ as required.

	For the induction step,  $t$ is not a leaf. Then $t=?(t_{1},t_{2},...,t_{p})$. 
	If $H[t]\neq 0$, then since $H$ is zero-consistent, we have for every $1\leq i\leq p$ that 
	$H[t_{i}]\neq 0$. So for every $1\leq j\leq p$ by induction hypothesis we have 
	$\mu(t_{j})=\classcoeff{t_{j}}{B}\cdot [t_{j}]$. Hence
	\begin{equation*}
	\begin{array}{rl}
	\mu(t)=&\mu_{?}(\mu(t_{1}),~\ldots~,\mu(t_{p}))=\\
	=&\mu_{?}(\classcoeff{t_{1}}{B}\cdot [t_1],~\ldots~,\classcoeff{t_{p}}{B}\cdot [t_{p}]) \\
	\end{array}
	\end{equation*}
	Therefore we have
	\begin{equation*}
	\mu(t)[j]=\sum_{j_{1},...,j_{p}\in [n]^{p}}\sigma^{j}_{j_{1},...,j_{p}}\cdot\classcoeff{t_{1}}{B} [t_1][j_{1}]\cdots\classcoeff{t_{p}}{B} [t_p][j_{p}]
	\end{equation*}
	Note that for every $j_{1},j_{2},...,j_{p}\neq \classind{t_{1}}{B},\classind{t_{2}}{B},...,\classind{t_{p}}{B}$ we have $\classcoeff{t_{1}}{B} [t_{1}][j_{1}]~\cdots~\classcoeff{t_{p}}{B} [t_{p}][j_{p}]=0$, thus
	\begin{align*}
	\mu(t)[j]&=\sigma^{j}_{\classind{t_{1}}{B},...,\classind{t_{p}}{B}}\cdot~ \classcoeff{t_{1}}{B} [t_{1}][\classind{t_{1}}{B}]~\cdots~\classcoeff{t_{p}}{B} [t_{p}][\classind{t_{p}}{B}]\\
	&=\sigma^{j}_{\classind{t_{1}}{B},...,\classind{t_{p}}{B}}\cdot \classcoeff{t_{1}}{B}\cdot \classcoeff{t_{2}}{B}~\cdots~\classcoeff{t_{p}}{B}
	\end{align*}
	From Alg.~\AlgExtractCMTA, and Lemma~\ref{lem:equal-prods-coef} it follows that 	
	$$\sigma^{j}_{\classind{t_{1}}{B},\classind{t_{2}}{B},...,\classind{t_{p}}{B}}=
	\left\{ \begin{array}{l@{\quad }l} 
	0 & \mbox{if } j\neq \classind{t}{B} \\
	\frac{\classcoeff{t}{B}}{\prod_{j=1}^{p}{B}\classcoeff{t_{p}}{B}} & \mbox{if } j=\classind{t}{B}
	\end{array}\right.$$
		Hence we obtain
	\begin{align*}
	\mu(t)[\classind{t}{B}]&=\sigma^{j}_{\classind{t_{1}}{B},\classind{t_{2}}{B},...,\classind{t_{p}}{B}}\cdot \classcoeff{t_{1}}{B}\cdot \classcoeff{t_{2}}{B}\cdots \classcoeff{t_{p}}{B} \\ &=\frac{\classcoeff{t}{B}}{\prod_{j=1}^{p}\classcoeff{t_{p}}{B}}\cdot {\prod_{j=1}^{p}\classcoeff{t_{p}}{B}}=\classcoeff{t}{B}
	\end{align*}
	Thus $\mu(t)=\classcoeff{t}{B}\cdot [t]$ as required.
	
	If $H[t]=0$ then $\sigma^{i}_{\classind{t_{1}}{B},\classind{t_{2}}{B},...,\classind{t_{p}}{B}}=0$ for every $i$, and we obtain
	$\mu(t)=0$  
	as required.
\end{proof}

Next we show that  rows in the basis get a standard basis vector. 

\begin{lemma}\label{lemmaBase}
	For every $b_{i}\in B$, $\mu(b_{i})=e_{i}$ where $e_{i}$ is the $i$'th standard basis vector.
\end{lemma}
\begin{proof}
	By induction on the height of $b_{i}$.
	For the base case, $b_{i}$ is a leaf, so $b=\sigma$ for $\sigma\in\Sigma_0$. By  Alg.~\AlgExtractCMTA we set $\sigma_{i}$ to be $1$ and $\sigma_{j}$ to be $0$ for every $j\neq i$, so $\mu(b_{i})=e_{i}$.
	
	For the induction step, $b_{i}$ is not a leaf. Note that by definition of the method $\Close$ (Alg.~\AlgCloseCMTA), all the children of $b_{i}$ are in $B$. So $b_{i}=\sigma(b_{i_{1}},b_{i_{2}},...,b_{i_{p}})$ for some base rows $b_{i_j}$'s. Let's calculate $\mu(b_{i})[j]$
	\begin{equation*}
	\mu(b_{i})[j]=\sum_{j_{1},j_{2},...,j_{p}\in [n]^{p}}\sigma^{j}_{j_{1},j_{2},...,j_{p}}\cdot\mu(b_{i_{1}})[j_{1}]\cdot\hdots\cdot\mu(b_{i_{p}})[j_{p}]
	\end{equation*}
	By the induction hypothesis, for every $1\leq j\leq p$ we have that $\mu(b_{i_{j}})[i_{j}]=1$, and $\mu(b_{i_{j}})[k]=0$ for $k\neq i_{j}$. So for every vector $j_{1},j_{2},...,j_{p}\neq i_{1},i_{2},...,i_{p}$ we obtain
	\begin{equation*}
	\mu(b_{i_{1}})[j_{1}]\cdot\hdots\cdot\mu(b_{i_{p}})[j_{p}]=0
	\end{equation*}
	For $j_{1},j_{2},...,j_{p}=i_{1},i_{2},...,i_{p}$ we obtain
	\begin{equation*}
	\mu(b_{i_{1}})[j_{1}]\cdot\hdots\cdot\mu(b_{i_{p}})[j_{p}]=1
	\end{equation*}
	Hence we have
	\begin{equation*}
	\mu(b_{i})[j]=\sigma^{i}_{i_{1},i_{2},...,i_{p}}
	\end{equation*}
	By Alg.~\AlgExtractCMTA we have that $\sigma^{i}_{i_{1},i_{2},...,i_{p}}=1$ and $\sigma^{j}_{i_{1},i_{2},...,i_{p}}=0$ for $j\neq i$, so $\mu(b_{i})[i]=1$ and $\mu(b_{i})[j]=0$ for $j\neq i$. Hence $\mu(b_{i})=e_{i}$ as required.
\end{proof}

The next lemma states 
for a tree $t=\sigma(b_{i_{1}},b_{i_{2}},...,b_{i_{p}})$ with
children in the basis, 
if $t\treesColin{\alpha}{H}b_i$ then  
 $\mu(t)=\alpha\cdot  e_{i}$ where $e_{i}$ 
 is the $i$'th standard basis vector.

\begin{lemma}\label{lemmaExtension}
	Let $t=\sigma(b_{i_{1}},b_{i_{2}},...,b_{i_{p}})$, s.t. ${b_{i_{j}}\in B}$ for ${1\leq j\leq p}$. Assume ${H[t]=\alpha\cdot H[b_{i}]}$ for some $i$. Then $\mu(t)=\alpha\cdot  e_{i}$.
\end{lemma}
\begin{proof}
	If $t=\sigma$ is a leaf, then by definition we have $\sigma_{i}=\alpha$ and $\sigma_{j}=0$ for $j\neq i$, so $\mu(t)=\alpha\cdot e_{i}$.
	Otherwise, $t$ isn't a leaf. Assume $t=\sigma(b_{i_{1}},b_{i_{2}},...,b_{i_{p}})$. We thus have
	\begin{equation*}
	\mu(t)[j]=\sum_{j_{1},j_{2},...,j_{p}\in [n]^{p}}\sigma^{j}_{j_{1},j_{2},...,j_{p}}\cdot\mu(b_{i_{1}})[j_{1}]\cdot\hdots\cdot\mu(b_{i_{p}})[j_{p}]
	\end{equation*}
	By Lemma \ref{lemmaBase} we have that $\mu(b_{i_{j}})=e_{i_{j}}$ for $1\leq j\leq p$, hence using a similar technique to the one used in the proof of Lemma \ref{lemmaBase} we obtain that for every $1\leq j\leq p$:
	\begin{equation*}
	\mu(t)[j]=\sigma^{j}_{i_{1},i_{2},...,i_{p}}
	\end{equation*}
	By Alg.~\AlgExtractCMTA we have that $\sigma^{j}_{i_{1},i_{2},...,i_{p}}=\alpha$ for $i=j$ and $\sigma^{j}_{i_{1},i_{2},...,i_{p}}=0$ for $i\neq j$, thus $\mu(t)=\alpha\cdot  e_{i}$ as required.
\end{proof}

The following lemma generalizes the previous lemma to any tree $t\in T$.

\begin{lemma}\label{lem:Correctness}
	Let $H$ be a closed consistent sub-matrix of the Hankel Matrix. Then for every $t\in T$ s.t. $H[t]=\alpha\cdot H[b_{i}]$ we have $\mu(t)=\alpha\cdot e_{i}$
\end{lemma}
\begin{proof}
	By induction on the height of $t$.
	For the base case $t$ is a leaf, and the claim holds by Lemma \ref{lemmaExtension}.
	Assume the claim holds for all trees of height at most $h$. Let $t$ be a tree of height $h$. Then $t=\sigma(t_{1},t_{2},...,t_{p})$. Since $T$ is prefix-closed, for every $1\leq j\leq p$ we have that $t_{j}\in T$. And from the induction hypothesis for every $1\leq j\leq p$ we have that $\mu(t_{j})=\alpha_{j}\cdot e_{i_{j}}$. Hence
	\begin{align*}
	\mu(t)&=\mu(\sigma(t_{1},t_{2},...,t_{p}))=\mu_{\sigma}(\mu(t_{1}),\mu(t_{2}),...,\mu(t_{p}))\\
	&=\mu_{\sigma}(\alpha_{1}\cdot e_{i_{1}},\alpha_{2}\cdot e_{i_{2}},...,\alpha_{p}\cdot e_{i_{p}})\\
	&=\prod_{j=1}^{p}\alpha_{j} \cdot \mu_{\sigma}(e_{i_{1}},e_{i_{2}},...,e_{i_{p}})
	\end{align*}
	Let $t'=\sigma(b_{i_{1}},b_{i_{2}},...,b_{i_{p}})$. From Lemma \ref{lemmaBase} we have
	\begin{align*}
	\mu(t)&=\prod_{j=1}^{p}\alpha_{j}\cdot\mu_{\sigma}(e_{i_{1}},e_{i_{2}},...,e_{i_{p}})\\
	&= \prod_{j=1}^{p}\alpha_{j}\cdot \mu_\sigma(\mu(b_{i_{1}}),\mu(b_{i_{2}}),...,\mu(b_{i_{p}}))\\
	&=\prod_{j=1}^{p}\alpha_{j}\cdot\mu(t')
	\end{align*}
	Since the table is consistent, we know that for each $1\leq j\leq p$ and $c\in C$:
	\begin{align*}
	H[\sigma(t_{1},t_{2},...,t_{j-1},t_{j},t_{j+1},...,t_{p})][c]=
	\alpha_{j}\cdot H[\sigma(t_{1},t_{2},...,t_{j-1},b_{i_{j}},t_{j+1},...,t_{p})][c]
	\end{align*}
	We can continue using consistency to obtain that
	\begin{align*}
	H[t][c]&=H[\sigma(t_{1},t_{2},...,t_{p})][c]\\
	&=\prod_{j=1}^{p}\alpha_{j}\cdot H[\sigma(b_{1},b_{2},...,,b_{p})][c]\\
	&=\prod_{j=1}^{p}\alpha_{j}\cdot H[t'][c]
	\end{align*}
	Thus $H[t]=\prod_{j=1}^{p}\alpha_{j}\cdot H[t']$. Let $\beta=\prod_{j=1}^{p}\alpha_{j}$, then $t\treesColin{\beta}{H} t'$. Let $b$ be the element in the base s.t. $t'\treesColin{\alpha}{H} b_{i}$. From Lemma \ref{lemmaExtension} we have that $\mu(t')=\alpha\cdot e_{i}$. Therefore $\mu(t)=\beta\cdot\alpha\cdot e_{i}$. 
	We have $\mu(t)=\beta\cdot\alpha\cdot e_{i}$ and $t\treesColin{\alpha\cdot\beta}{H} b_{i}$. Therefore the claim holds.
\end{proof}

We are now ready to show that for every tree $t\in T$ and context $c\in C$ the obtained CMTA
agrees with the observation table.

\begin{lemma}
	For every $t\in T$ and for every $c\in C$ we have that $\mathcal{A}(\contextConcat{c}{t})=H[t][c]$.
\end{lemma}
\begin{proof}
	Let $t\in T$. Since the table is closed, there exists $b_{i}\in B$ such that $t\treesColin{\alpha}{H} b_{i}$ for some $\alpha\in\mathbb{R}_+$.
	The proof is by induction on the depth of $\context$ in $c$. 
	For the base case, the depth of $c$ is $1$, so $c=\context$, and by Lemma \ref{lem:Correctness} we have that $\mu(\contextConcat{c}{t})=\mu(t)=\alpha \cdot e_{i}$. Therefore $\mathcal{A}(t)=\alpha\cdot e_{i}\cdot\lambda$.  By Alg.~\AlgExtractCMTA we have that $\lambda[i]=H[b_{i}][\context]$. Thus $\mathcal{A}(t)=\alpha\cdot H[b_{i}][\context]=H[t][\context]$ as required. 
	
	For the induction step, let $c$ be a context such that the depth of $\context$ is $h+1$. Hence $c=\contextConcat{c'}{\sigma(t_{1},t_{2},...,t_{i-1},\context,t_{i},...,t_{p})}$ for some trees $t_j\in T$, and some context $c'$  of depth  $h$.  For each $1\leq j\leq p$, let $b_{i_{j}}$ be the element in the base, s.t. $t_{j}\equiv_{H} b_{i_{j}}$, with co-efficient $\alpha_{j}$. Let $b$ be the element in the base s.t. $t\equiv_{H} b$ with coefficient $\alpha$. Let $\widetilde{t}$ be the tree:
	\begin{equation*}
	\widetilde{t}=\sigma(b_{i_{1}},b_{i_{2}},...,b_{i_{k-1}},b,b_{i_{k}},...,b_{i_{p}})
	\end{equation*}
	Note that $\widetilde{t}\in \Sigma(B)$ and hence $\widetilde{t}\in T$. From the induction hypothesis, we obtain:
	\begin{equation*}
	\mathcal{A}(\contextConcat{c'}{\widetilde{t}})=H[\widetilde{t}][c']
	\end{equation*}
	Since the table is consistent, we have:
	\begin{align*}
	H[t][c]&=H[\sigma(t_{i_{1}},t_{i_{2}},...,t_{i_{k-1}},t,t_{i_{k}},...,t_{i_{p}})][c']=\alpha\cdot\prod_{i=1}^{p}\alpha_{i}\cdot H[\widetilde{t}][c']
	\end{align*}
	Let $\beta=\alpha\cdot\prod_{i=1}^{p}\alpha_{i}$.
	By definition of $\aut{A}$ we have: 
	\begin{align*}
	\aut{A}(\contextConcat{c'}{\sigma(b_{i_{1}},b_{i_{2}},...,b_{i_{k-1}},b,b_{i_{k}},...,b_{i_{p}})})=
	\mu(\contextConcat{c'}{\sigma(b_{i_{1}},b_{i_{2}},...,b_{i_{k-1}},b,b_{i_{k}},...,b_{i_{p}})}))\cdot\lambda
	\end{align*}
	Since each $t_{i_{j}}$ is in $T$, from Proposition \ref{lem:Correctness} we have that ${\mu(t_{i_{j}})=\alpha_{j}\cdot b_{i_{j}}}$, and that \linebreak[4]
  ${\mu(t)=\alpha\cdot\mu(b)}$.
	Let $\hat{t}=\sigma(t_{i_{1}},t_{i_{2}},...,t_{i_{k-1}},t,t_{i_{k}},...,t_{i_{p}}))$.
	 So
	\begin{align*}
	\mu(\hat{t}) 
	&=\mu(\sigma(t_{i_{1}},t_{i_{2}},...,t_{i_{k-1}},t,t_{i_{k}},...,t_{i_{p}})))\\
	&=\mu_{\sigma}(\alpha_{1}\cdot b_{i_{1}},...,\alpha_{k-1}\cdot b_{i_{k-1}},\alpha\cdot b,\alpha_{k}\cdot b_{i_{k}},...,\alpha_{p}\cdot b_{i_{p}})\\
	&=\alpha\cdot\prod_{j=1}^{p}\alpha_{i}\cdot\mu_{\sigma}(b_{i_{1}},...,b_{i_{k-1}},b, b_{i_{k}},...,b_{i_{p}})
	=\beta\cdot\mu(\widetilde{t})
	\end{align*}
	By Lemma \ref{replacementLemma} we have that $$\mu(\contextConcat{c'}{\sigma(t_{i_{1}},t_{i_{2}},...,t_{i_{k-1}},t,t_{i_{k}},...,t_{i_{p}})})=\beta\cdot\mu(\contextConcat{c'}{\widetilde{t}})$$
	Hence
	\begin{align*}
	\mathcal{A}(\contextConcat{c}{t})&=\mathcal{A}(\contextConcat{c'}{\sigma(t_{i_{1}},t_{i_{2}},...,t_{i_{k-1}},t,t_{i_{k}},...,t_{i_{p}})})\\
	&=\mu(\contextConcat{c'}{\sigma(t_{i_{1}},t_{i_{2}},...,t_{i_{k-1}},t,t_{i_{k}},...,t_{i_{p}})})\cdot\lambda
	=\beta\cdot\mu(\contextConcat{c'}{\widetilde{t}})\cdot\lambda
	\end{align*}
	Note that all the children of $\widetilde{t}$ are in $B$, and so $\widetilde{t}\in T$. Hence, from the induction hypothesis we have
	\begin{equation*}
	H[\widetilde{t}][c']=\mathcal{A}(\contextConcat{c'}{\widetilde{t}})=\mu(\contextConcat{c'}{\widetilde{t}})\cdot\lambda
	\end{equation*}
	Thus
	\begin{equation*}
	\mathcal{A}(\contextConcat{c}{t})=\beta\cdot H[\widetilde{t}][c']=H[t][c]
	\end{equation*}
	as required.
\end{proof}

\section{Discussion}
The quest to learn probabilistic automata and grammars is still ongoing. 
Because of the known hardness results some restrictions need to be applied.
Recent works include an  \lstar\ learning algorithm for MDPs~\cite{TapplerA0EL21}
(here the assumption is that states of the MDPs generate an observable output that allows identifying the current state based on the generated input-output sequence),
a passive learning algorithm for a subclass of PCFGs obtained by imposing several structural restrictions~\cite{ClarkF20,pmlr-v153-clark21a}, 
and using PDFA learning to obtain an interpretable model 
of practically black-box models such as recurrent neural networks~\cite{WeissGY19}.
 
 We have presented an algorithm for learning structurally unambiguous PCFGs from a given black-box language model using structured membership and equivalence queries.
To our knowledge this is the first algorithm provided for this question. 
Following the motivation of~\cite{WeissGY19}, the present work offers obtaining intrepretable models also in cases where the studied
object exhibits non-regular (yet context-free) behavior.
For future work, we think that improving our method to be more noise-tolerant would make the algorithm able to learn complex regular and context-free grammars from recurrent neural networks.

\commentout{
\paragraph{New Related Work}\dana{move}

The paper provides an

Probabilistic automata have also been recently investigated in the passive learning paradigm~\cite{ClarkF20,}.
These work consider anchored PCFGs which pose three restrictions on PCFGs: .........
They prove convergence with (approximately) correct probabilities to an isomorphic PCFG is guaranteed (no complexity analysis is given).\dana{There are probably more references on passive learning PCFGs}
}

\bibliographystyle{alphaurl}
\bibliography{bib}

\end{document}